\documentclass[lettersize,journal]{IEEEtran}
\usepackage{amsmath,amsfonts}
\usepackage{algorithmic}
\usepackage[ruled,lines numbered,vlined,boxed]{algorithm2e}
\usepackage{array}
\usepackage[caption=false,font=normalsize,labelfont=sf,textfont=sf]{subfig}
\usepackage{textcomp}
\usepackage{stfloats}
\usepackage{hyperref}
\usepackage{utfsym}
\usepackage{url}
\usepackage{pifont}
\usepackage{verbatim}
\usepackage{graphicx}
\usepackage{cite}
\newtheorem{theorem}{Theorem}[section]

\newtheorem{definition}{Definition}

\begin{document}

\title{Efficient and Privacy-Preserving Distribution Statistics Analytics on Mobile Spatial Data}

\author{Xuhao Ren, Chuan Zhang,~\IEEEmembership{Member,~IEEE}, Mingyang Zhao, Meng Li,~\IEEEmembership{Senior Member,~IEEE}, \\Liehuang Zhu,~\IEEEmembership{Senior Member,~IEEE}, Bin Xiao,~\IEEEmembership{Fellow,~IEEE}
        % <-this % stops a space
\thanks{This work is supported by the National Natural Science Foundation of China (Grant No. 62472032).}% <-this % stops a space
\thanks{Xuhao Ren, Chuan Zhang, and Liehuang Zhu are with the School of Cyberspace Science and Technology, Beijing Institute of Technology, Beijing 100081, China, and also with the Shandong Key Laboratory of Energy Industry Internet Big Data Technology, Jinan 250003, Shandong, China. E-mail:
\{xuhaor, chuanz, liehuangz\}@bit.edu.cn.}
\thanks{Mingyang Zhao and Bin Xiao are with the Department of Computing, The Hong Kong Polytechnic University, Hong Kong. Email: 25019897r@connect.polyu.hk and b.xiao@polyu.edu.hk.}
\thanks{Meng Li is with the School of Computer Science and Information Engineering, Hefei University of Technology, Hefei, Anhui 230601, China. Email: mengli@hfut.edu.cn.}
\thanks{Chuan Zhang is the corresponding author.}}

\markboth{IEEE Transactions on Information Forensics and Security}%
{Shell \MakeLowercase{\textit{et al.}}: A Sample Article Using IEEEtran.cls for IEEE Journals}


\maketitle

\begin{abstract}
With the rapid development of mobile computing technology, massive amounts of spatial data are continuously generated from various mobile terminals and sensing devices, such as smartphones, connected vehicles, and drones. Performing efficient distributed statistical analysis on this data is crucial for real-time mobile computing applications. However, the constrained and dynamic nature of mobile environments exacerbates the privacy challenge: centralizing sensitive data for analysis risks severe privacy leaks, while existing privacy-preserving techniques often introduce excessive overhead or inaccuracies.
In this paper, we design, implement, and evaluate the first system that supports efficient and privacy-preserving distribution statistics analysis for mobile spatial data. First, we propose $\mathsf{eSpat\mbox{-}B}$, which leverages two non-colluding servers and a newly designed improved distributed point functions (DPF) with octree partitioning.
Furthermore, considering the frequent updates of spatial data, we propose another more efficient scheme, $\mathsf{eSpat+}$. The core idea of this scheme is to utilize a K-Dimensional tree for spatial partitioning, combine it with incremental DPF for performing statistics analysis, and design an efficient update algorithm. Security analysis demonstrates that our schemes effectively protect data privacy throughout the statistical process. Extensive experiments on real-world trajectory datasets demonstrate that the proposed schemes significantly outperform existing approaches, reducing computation overhead by up to $1.2\times$ and communication overhead by up to $20\times$ while maintaining $100\%$ statistical accuracy.
\end{abstract}


\section{Introduction}
\IEEEPARstart{T}{he} proliferation of mobile and IoT devices has led to an unprecedented growth of spatial data—characterized by latitude, longitude, and often altitude—generated continuously from smartphones, connected vehicles, drones, and other moving entities~\cite{xue2026TNSE}. Such three-dimensional spatial data has become a fundamental resource for a broad range of applications, including intelligent transportation systems, UAV-based monitoring, environmental sensing, and urban management~\cite{zhao2026pods}. Analyzing the distribution of such data provides valuable insights into movement patterns, density hotspots, and resource utilization across geographic regions \cite{wandelt2023aerial}. The core task of spatial distribution statistics analysis can be defined as: given a stream or set of mobile spatial data points (e.g., vehicle locations and UAV trajectories) and a predefined spatial partitioning (grid or region), to efficiently and accurately compute statistical summaries—such as density, coverage, or frequency—over these partitions. For instance, a mobility service may analyze vehicle distribution to optimize fleet allocation; a city operator might assess crowd density across zones for public safety management~\cite{xue2024TMC}.

However, conducting such statistical analysis typically requires collecting data at a central server, which raises serious privacy concerns in mobile contexts \cite{zhang2024Jsac,zhang2023pota,xie2026flexible,liang2025secpq}. Mobile spatial data often contains sensitive trajectory and location information, and its exposure could lead to tracking, profiling, or other security breaches. A straightforward solution is to encrypt data before uploading, but this introduces significant computational and communication overhead. Moreover, dynamic and streaming nature of mobile data demands efficient handling of frequent updates—a challenge overlooked by many static privacy-preserving methods. Inaccurate or delayed statistics analysis in mobile scenarios can directly impact system performance and safety. For example, inefficient routing in vehicular networks may increase travel time and energy consumption, while imprecise UAV location analytics could raise collision risks or lead to violations of no-fly zones. Therefore, achieving efficient, accurate, and privacy-preserving distribution statistics analysis over mobile spatial data in practical, resource-constrained environments remains a critical challenge.

\begin{table*}[]
\centering
\caption{Comparison with related works}
\label{tab:comparison}
\renewcommand{\arraystretch}{1.3}
\setlength{\tabcolsep}{10pt}
\begin{tabular}{|c||c|c|c|c|c|}
\hline
\hline
Schemes      & Security Primitives     & 3D Statistics & Privacy & High Efficiency & Accuracy \\ \hline
PPRQ~\cite{9453150}         & Symmetric encryption    & $\usym{2717}$              & $\usym{2713}$      & $\usym{2717}$                & $\usym{2713}$       \\ \hline
PPDA~\cite{wu2023robust}         & Homomorphic encryption  & $\usym{2717}$              & $\usym{2713}$      & $\usym{2717}$                & $\usym{2713}$       \\ \hline
DDP~\cite{bell2022distributed}          & Differential privacy    & $\usym{2717}$              & $\usym{2713}$      & $\usym{2713}$              & $\usym{2717}$         \\ \hline
Prio~\cite{corrigan2017prio}         & Multi-party computation & $\usym{2717}$              & $\usym{2713}$      & $\usym{2717}$                & $\usym{2713}$       \\ \hline
Sketch~\cite{melis2015efficient}       & N/A                     & $\usym{2717}$              & $\usym{2717}$        & $\usym{2713}$              & $\usym{2717}$         \\ \hline
Standard DPF~\cite{boyle2016function} & Function secret sharing & $\usym{2717}$              & $\usym{2713}$      & $\usym{2717}$                & $\usym{2713}$       \\ \hline
\hline
Our scheme   & Function secret sharing & $\usym{2713}$            & $\usym{2713}$      & $\usym{2713}$              & $\usym{2713}$       \\ \hline
\end{tabular}
\vspace{-0.1in}
\end{table*}

% \vspace{.07in}
Some privacy-preserving techniques can be applied to spatial distribution statistics analysis, including searchable encryption (SE)~\cite{zhang2024Jsac,song2025high}, privacy-preserving range queries (PPRQ)~\cite{Tong2024TIFS,wu2025hex}, and private statistics~\cite{assouline2023weighted,tao2025differentially}. SE enables querying encrypted data while preserving data confidentiality. For example, Chamani et al.~\cite{chamani2021multi} proposed a blockchain-based multi-user searchable encryption scheme, while Le et al.~\cite{le2024muses} further concealed statistical leakage in multi-user settings. In the PPRQ domain, Tong et al.~\cite{Tong2024TIFS} leveraged distributed point functions (DPFs) and Bloom filters to support privacy-preserving range queries. However, when applied to spatial distribution statistics, these approaches typically require generating query-specific indexes and matching individual records, resulting in substantial computation and communication overhead. Private statistics techniques provide another promising direction. For instance, Boneh et al.~\cite{boneh2021lightweight} proposed a secure two-server framework for computing heavy hitters, and Mouris et al.~\cite{mouris2024plasma} further improved communication efficiency using a three-server architecture. In addition, local differential privacy (LDP)-based approaches~\cite{hong2022collecting,cunningham2021real} protect user locations through perturbation mechanisms, but inevitably introduce statistical errors due to noise injection. Overall, existing approaches either incur significant overhead or sacrifice accuracy, and most are not designed for efficient privacy-preserving distribution statistics over dynamic three-dimensional spatial data.

Based on the above requirements, we design the first system that supports efficient and privacy-preserving distribution statistics analysis for spatial data, which comprises two schemes. To be specific, we design the scheme, $\mathsf{eSpat\mbox{-}B}$, based on improved DPF, where DPF is combined with octree partitioning and regions are divided using Gray code to achieve privacy protection. To further improve efficiency, we also present an advanced scheme, $\mathsf{eSpat+}$, which is designed using K-Dimensional (KD)-tree encoding, incremental DPF, and an efficient update algorithm. Table~\ref{tab:comparison} shows the comparison of our scheme with previous works in terms of requirements. In summary, the main contributions of this paper are:
\begin{itemize}
    \item We explore and analyze the challenges and requirements for privacy-preserving spatial distribution statistics analysis, and propose a system that ensures spatial data privacy while achieving efficient and accurate statistics analysis.
    \item We propose a basic spatial data distribution statistics analysis scheme, $\mathsf{eSpat\mbox{-}B}$. Specifically, $\mathsf{eSpat\mbox{-}B}$ encodes space using Gray code and combines the DPF with octree partitioning to design an improved DPF algorithm. This scheme allows for the generation of key shares that are distributed across two non-colluding servers and ensures data privacy while enabling efficient distribution of statistics. 
    \item To further improve the performance of spatial distribution statistics analysis and update efficiency, we design an advanced scheme called $\mathsf{eSpat+}$. Specifically, we encode space using a KD-tree and leverage incremental DPF to protect privacy. In addition, considering the frequent updates of spatial data, we design an efficient update algorithm that reduces overhead by optimizing updates within the same parent region.
    \item We perform extensive security analysis to validate the robustness of our schemes and demonstrate their efficiency through extensive experiments, confirming their practical applicability and high performance.
\end{itemize}

The rest of the paper is structured as follows. Section \ref{sec:2} presents the background of our scheme. We provide the related work in Section \ref{sec:7}. In Section \ref{sec:3}, we give a problem formulation of our scheme, including the system model, threat model, definitions of our work, and design goals. We give the scheme details in Section \ref{sec:4}, followed by a security analysis in Section \ref{sec:5}. Next, performance evaluation is discussed in Section \ref{sec:6}. Finally, we conclude this paper in Section \ref{sec:8}.

\section{Related Work}\label{sec:7}
In this part, we introduce some existing technologies that can be used to achieve privacy-preserving spatial distribution statistics analysis.

\subsection{Searchable Encryption}
Searchable encryption (SE) enables efficient retrieval over encrypted data while preserving the confidentiality of both stored data and search queries~\cite{zhang2017searchable,qiu2020survey}. Since the introduction of public-key encryption with keyword search (PEKS) by Boneh et al.~\cite{boneh2004public}, numerous SE schemes have been proposed to improve search functionality and efficiency. For example, Curtmola et al.~\cite{curtmola2006searchable} introduced an inverted-index-based SE construction to accelerate query processing, while Tong et al.~\cite{tong2022privacy} combined Bloom filters and cryptographic techniques to support privacy-preserving range queries. More recently, Liang et al.~\cite{liang2023privacy} further improved retrieval efficiency through an enhanced Bloom filter design.

SE can be applied to privacy-preserving spatial statistics by encrypting spatial records and retrieving data points within a queried region before performing aggregation. However, spatial distribution statistics typically require processing large volumes of dynamically updated spatial data. Consequently, SE-based approaches must first retrieve and match individual records prior to aggregation, resulting in considerable computation and communication overhead. Such inefficiencies become particularly pronounced in large-scale three-dimensional spatial datasets with frequent location updates.

\subsection{Privacy-Preserving Range Query}
Privacy-preserving range query (PPRQ) techniques enable users to retrieve records within a query range while protecting both the query and the underlying data. Existing approaches employ various cryptographic primitives, including secure kNN, homomorphic encryption, Bloom filters, and hidden vector encryption, to support secure range queries over encrypted datasets~\cite{tong2021privacy,song2021privacy,gong2022efficient,guan2022achieving,wang2023efficient}. More recently, Tong et al.~\cite{tong2022privacy} and Miao et al.~\cite{miao2023efficient} proposed efficient range query schemes for multidimensional spatial data using symmetric encryption and homomorphic encryption techniques.

PPRQ techniques can be applied to spatial distribution statistics by retrieving all records within a target region and subsequently performing aggregation. However, such approaches fundamentally follow a ``query-then-aggregate'' paradigm, requiring individual record matching before statistical computation. As a result, they generate a large number of intermediate results and incur substantial computation and communication overhead, particularly for large-scale spatial datasets and frequent statistical queries.

\subsection{Private Data Statistics}
If the server needs to calculate a collection of all client strings, the involved parties can utilize a dual-server mixing network~\cite{chaum1981untraceable}. In this setup, each client encrypts her string with onion encryption to two servers, who then shuffle and decrypt the group of strings, respectively. 
Employing verifiable shuffles~\cite{neff2001verifiable} is essential to prevent any misconduct by the servers. An alternative approach involves employing a universally secure two-party computation~\cite{fanti2015building, lu2013distributed} for the RAM program. In this method, each client shares an extra secret component of its input string with every server. Subsequently, the servers execute a universally secure multi-party computation on the RAM program. This computation processes inputs (one string from each client) and determines the quantity of heavy hitters present. The count-min sketch~\cite{corrigan2017prio} is a data structure designed for identifying approximate heavy hitters within streaming algorithms. In a study by Melis et al.~\cite{melis2015efficient}, it was demonstrated that secure aggregation methods could enable each client to anonymously input its string into a data structure. In cases where the specific heavy hitters are not predetermined, such as in our scenario, a series of $n$-like counting data structures (with each client managing an $n$-bit string) can be employed to tally the heavy hitters. Obviously, the above private data statistics analysis can only process one-dimensional data and cannot process three-dimensional data, such as UAV locations.

\section{Background}\label{sec:2}
\subsection{Gray Code}
Gray Code \cite{bitner1976efficient,wang2020search}, also known as Reflected Binary Code, is a special form of binary encoding, characterized by the fact that only one binary bit changes between adjacent codes. An $ d$-dimensional Gray code can be represented as shown below, with $'|'$ being the concatenation operator and $G_d$ representing the vector of a Gray code instance at step $d$. For example, $G_1= (0,1)$, and $d>1$, $G_{d} = (g_1, g_2 ,\ldots, g_{2^d}),G_{d+1} = (0|g_1, 0|g_2, \ldots, 0|g_{2^d}, \ldots, 1|g_{2^d}, \ldots, 1|g_1)$. As shown in Fig.~\ref{fig: quad-gray}, this is a two-dimensional spatial code.

\textcolor{black}{In a $D \times D$ grid, the Gray code needed to represent all cells has a length of $2^{\lceil \log_2 D\rceil}$. For a spatial point $U$, we represent the Gray code of $U$ as $g_U\gets Gray(U)$. To perform statistical analysis in three-dimensional space, we adopt the encoding method outlined in this paper: $G_3=(000,001,011,010,110,111,101,100)$. In the quad-tree partition, each internal node divides the current spatial region into four child regions, and each leaf node represents a final spatial subregion; the Gray code is used to encode the root-to-leaf path, thereby assigning a unique Gray-code-based index to each leaf region.}

\begin{figure}
    \centering
    \includegraphics[width=0.48\textwidth]{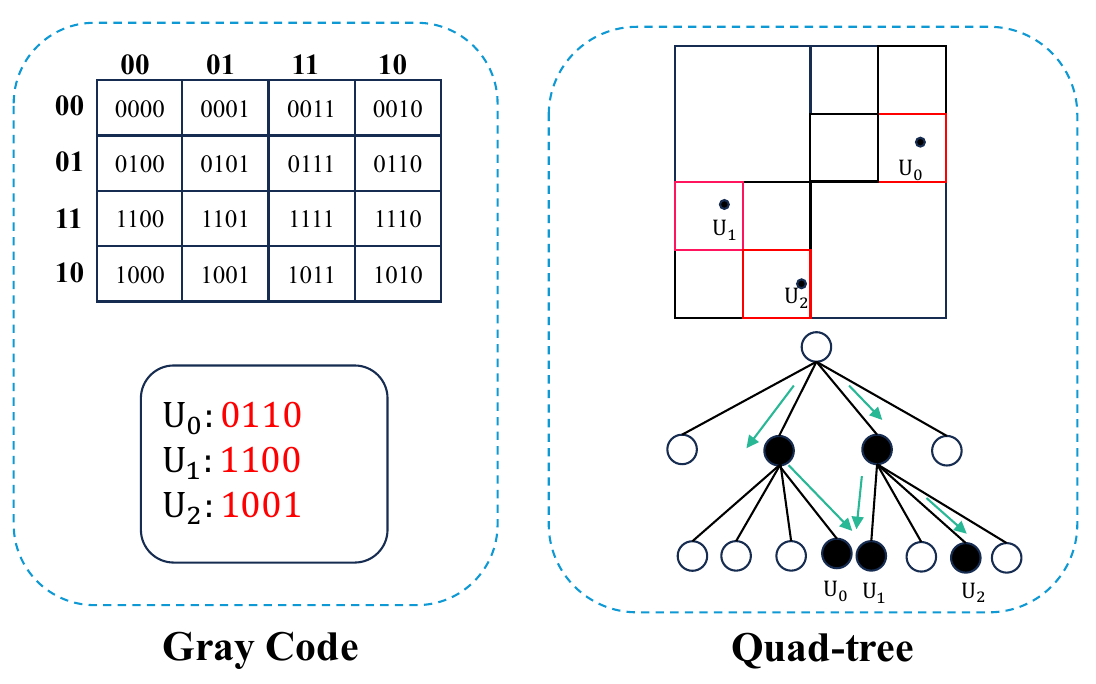}
    \caption{Example of Quad-tree and Gray code. The left part depicts a Quad-tree, a hierarchical data structure that recursively partitions a two-dimensional spatial region into four child nodes. The right part shows the Gray code assigned to the leaf cells of the Quad-tree. }
    \label{fig: quad-gray}
    \vspace{-0.15in}
\end{figure}

\subsection{Distributed Point Functions (DPF)}\label{sec:DPF}
The DPF introduced in~\cite{boyle2016function} is modeled by a function $f_{x,y}$ whose output is non-zero \emph{only} when the input matches the fixed string $x$; for every other input, the result is identically zero. Concretely, a two-party DPF over a finite field $\mathbb{F}$ is realized by the following pair of algorithms:
\begin{itemize}
    \item $\mathtt{DPF}.\mathtt{Gen}(1^{\lambda}, x, y) \rightarrow (k_0, k_1)$: on input a security parameter $\lambda$, an index string $x\in\{0,1\}^n$, and a payload value $y\in\mathbb{F}$, this routine produces two correlated keys $k_0$ and $k_1$.
    \item $\mathtt{DPF}.\mathtt{Eval}(i, k_b, x')\rightarrow y_b'$: given party identifier $b\in\{0,1\}$, a key $k_b$, and any $x'\in\{0,1\}^n$, the algorithm returns the share $y_b'$.
\end{itemize}

% The correctness requirement for the two-party DPF is expressed as follows.
% \begin{equation}\label{DPF.eq}
% \mathtt{DPF}.\mathtt{Eval}(k_0, x') + \mathtt{DPF}.\mathtt{Eval}(k_1, x') = 
% \begin{cases}
% y, & \text{if $x' = x$}\\
% 0, & \text{otherwise}
% \end{cases}.
% \end{equation}

% \begin{theorem}[Security of DPF~\cite{boyle2016function}]
%     When a party $P_0$ or $P_1$ is compromised by a PPT adversary $\mathcal{A}$, there is a PPT algorithm $SIM$ such that the outputs of the two experiments below cannot be distinguished computationally, given the leakage of the point function $f$'s input and output sizes, i.e., $\mathcal{L}_{DPF}(f) = (|x|, |y|):$
% \begin{itemize}
%     \item $REAL^{DPF}(1^\lambda)$:  Output $(DPF.Gen(1^\lambda, x, y)) \to (k_0, k_1)$;
%     \item $IDEAL^{DPF}(1^\lambda)$: Output $SIM(1^\lambda, \mathcal{L}^{DPF}(f))$.
% \end{itemize}
% \end{theorem}

% Notice that Eq.~\ref{DPF.eq} is computed in the finite field $\mathbb{F}$. The security property of the DPF suggests that an attacker who acquires either $k_0$ or $k_1$ (but not both) cannot learn any details about the point $\alpha$ or its corresponding value $\beta$.

\subsection{Incremental Distributed Point Functions}
A typical DPF~\cite{boyle2016function} effectively splits a non-zero vector of size $2^n$ at a specific point. The standard DPF divides the data into a binary tree structure where one leaf node stores a non-zero value $\beta$, with the remaining nodes holding zeros. This method entails distributing a unique value $\beta$ at each index. Conversely, the incremental DPF distributes values at every index prefix. (Refer to Fig. \ref{fig2})

Specifically, a two-party incremental DPF scheme, defined by a finite group $\mathbb{G}_1,\ldots,\mathbb{G}_n$, comprises two procedures:
\begin{itemize}
    \item $\mathtt{IDPF}.\mathtt{Gen}(1^{\lambda}, \alpha, \beta_1,\ldots,\beta_n) \rightarrow (k_0, k_1).$ Given a security parameter $\lambda$, a string $\alpha \in \{0, 1\}^n$ and values $\beta_1 \in \mathbb{G}_1,\ldots,\beta_n \in \mathbb{G}_n$, output two incremental DPF keys. 
    \item $\mathtt{IDPF}.\mathtt{Eval}(i, k_i, x) \rightarrow \mathbb{G}_l.$ For party index $b\in\{0,1\}$, an incremental DPF key $k_b$, and prefix $x\in\{0,1\}^l$, this algorithm outputs a secret-shared value located at index $x$. The computation is split into two internal routines: $\mathtt{EvalNext}$ advances the current state while producing the next share, and $\mathtt{EvalPrefix}$ delivers the corresponding share $y_b^{l}$.
\end{itemize}

\begin{figure}
    \centering
    \includegraphics[width=0.48\textwidth]{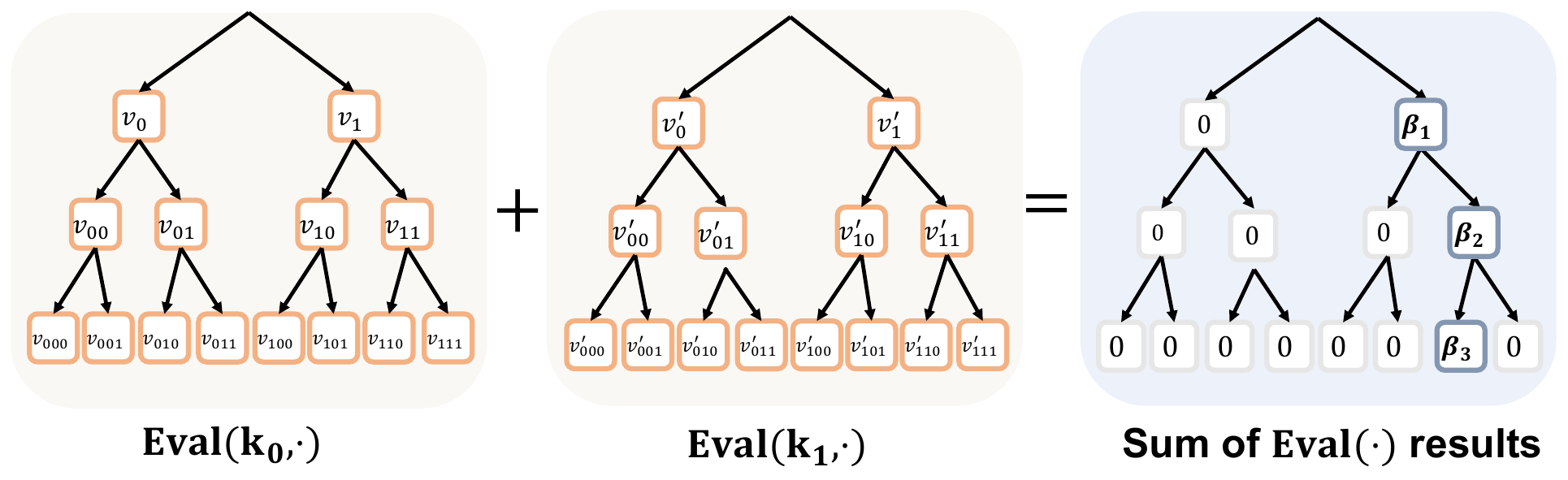}
    \caption{For instance, with a depth of $n=3$, a designate point $\alpha=110$, the values on the path are $\beta_1\in \mathbb{G}_1, \beta_2\in \mathbb{G}_2, \beta_3\in \mathbb{G}_3$ within specific finite groups $\mathbb{G}_1, \mathbb{G}_2, \mathbb{G}_3$, and the key generation is denoted as $\mathtt{Gen}(\alpha, \beta_1, \beta_2, \beta_3)\to (k_0, k_1)$.} 
    \label{fig2}
\end{figure}

For all $\alpha \in \{0,1\}^n$, output values $\beta_1, \ldots, \beta_n \in \mathbb{G}$, keys $(k_0, k_1)$, the following equation holds:
\begin{equation}
\mathtt{IDPF}.\mathtt{Eval}(k_0, x) + \mathtt{IDPF}.\mathtt{Eval}(k_1, x) = 
\begin{cases}
\beta_l, & \text{$x\in \alpha$}\\
0, & \text{otherwise}
\end{cases},
\end{equation}
where the $x$ has a length of $l$, and the operation is conducted within the finite group $\mathbb{G}_l$.

\subsection{KD-Tree}

KD-tree \cite{DBLP:journals/tvcg/ZhaoWZFXM22} is a tree structure designed for processing $k$-dimensional spatial data, where $k$ is the spatial dimension. Every node partitions the space along a specific dimension and allocates the data to its left and right child nodes. The partitioning method of the KD-tree is as follows.
\begin{itemize}
    \item During the construction process, the median of a dimension is usually selected as the partition point to ensure the balance of the tree.
    \item In each layer, the partitioning dimensions are selected in a cyclic order to ensure that the data in each dimension can be effectively partitioned.
\end{itemize}

Let us take the two-dimensional plane as an example. As shown in Fig. \ref{fig: kd-tree}, some points are randomly selected as the division criteria. The horizontal line is the x-axis. The left side of the x-axis is the point whose horizontal coordinate is less than the coordinate, and the right side is the point whose vertical coordinate is greater than the coordinate. The vertical line is the y-axis. The upper side of the y-axis is the point whose vertical coordinate is greater than the coordinate, and the lower side is the point whose vertical coordinate is less than the coordinate. 

\begin{figure}
    \centering
    \includegraphics[width=0.48\textwidth]{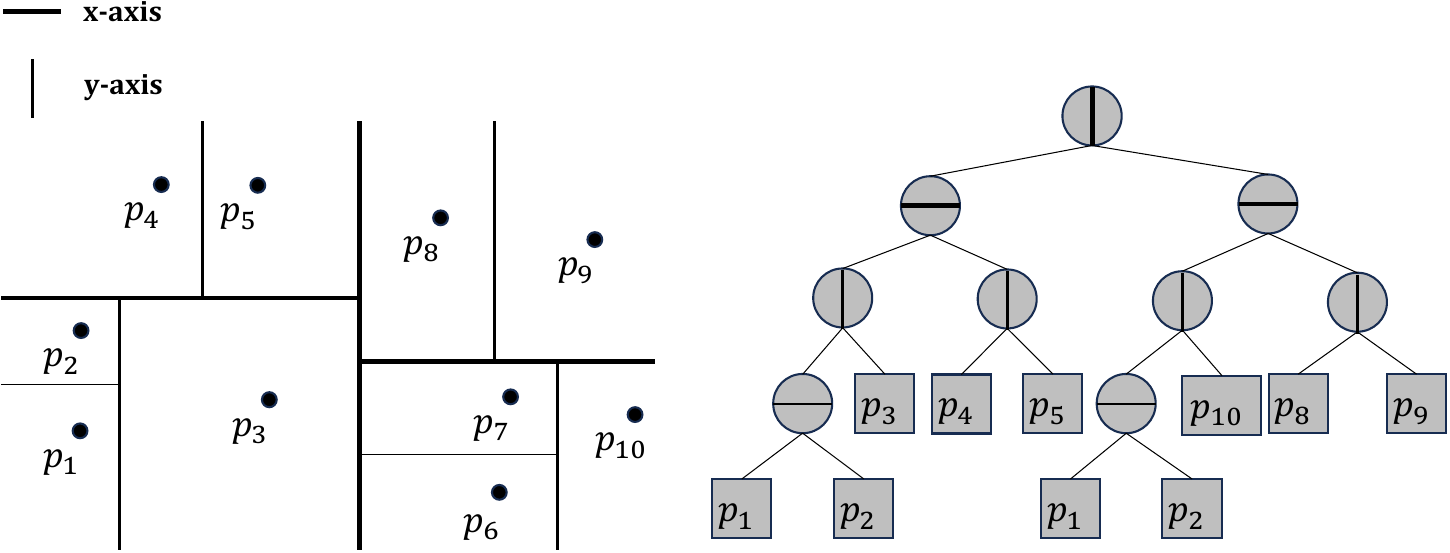}
    \caption{Example of KD-tree.}
    \label{fig: kd-tree}
    \vspace{-0.1in}
\end{figure}

\section{Problem Formulation}\label{sec:3}
In this section, we introduce the system and threat model examined in our work. It then outlines the definition of $\mathsf{eSpat}$.

\subsection{System Model and Threat Model}
\textbf{System model.} Our scheme consists of three entities: clients, cloud servers, and the requester.
\begin{itemize}
    \item \textbf{Clients.} Each client is responsible for uploading spatial data and utilizing an improved DPF to generate key shares that protect privacy. By using key sharing, the uploaded data can be securely processed by the servers without exposing the original spatial data.
    \item \textbf{Cloud servers.} The two servers receive key shares uploaded by the client and perform aggregation on the shares without accessing the spatial data, thereby ensuring privacy protection. After completing the aggregation, the servers send the results to the requester for the final synthesis of the distribution statistics.
    \item \textbf{Requester.} The requester receives the aggregated results from two servers and performs the final statistics analysis of these results to obtain the distribution of spatial data.
\end{itemize}

\textbf{Threat model.} \textcolor{black}{We assume that the client is honest, i.e., it accurately uploads real-time spatial data information as requested; the server and requester are assumed to be semi-honest. That is, although the server and requester follow the protocol to perform calculations, they may try to infer specific data from the received information. To ensure privacy, we consider that the two servers do not collude and therefore process the data independently~\cite{boneh2021lightweight, damgaard2024differentially,balle2025hash}. This assumption is commonly made in two-server privacy-preserving computation and can be instantiated by deploying the two servers under independent administrative domains.\footnote{To prevent collusion, incentive mechanisms based on game theory~\cite{dong2017betrayal} can be designed to penalize malicious cooperation and encourage honest behavior, which we leave as future work.}.}

\subsection{Definitions}
\textcolor{black}{Let $\mathcal{D}=\{p_1,p_2,\ldots,p_n\}$ denote a set of spatial data records, where each $p_i$ represents a location point in the spatial domain. Given a query region $R$, the goal of spatial distribution statistics in this paper is to compute the number of data records located within $R$. Formally, the output is a single count
\[
c_R = |\{p_i \in \mathcal{D} \mid p_i \in R\}|.
\]
Thus, in this work, the term ``distribution'' specifically refers to count-based spatial statistics over queried regions. Each query returns the number of records in the specified region, rather than a density, frequency, or a count vector over all regions. We focus on the count result because it is the fundamental statistic for spatial distribution analysis and directly reflects the number of records within the queried region. Density or frequency can be further derived from the count when the region size or total number of records is available, while returning a count vector over all regions would provide broader information than required by a single query.}

Based on the above model, we define an efficient spatial data distribution statistics analysis scheme called $\mathsf{eSpat}$. In $\mathsf{eSpat}$, we first partition the statistical space into multiple sub-areas, each assigned a unique code. Clients map their spatial data to these encoded sub-areas and generate key shares $(k_0,k_1)$. The server aggregates these shares via secure distributed computation and returns partial results to the requester, who then merges them to obtain the complete statistical distribution.

\begin{definition}[$\mathsf{eSpat}$]
    $\mathsf{eSpat}$ consists of three algorithms:
    \begin{itemize}
        \item $\mathtt{Setup()}$: This function generates a pseudo-random generator and a pseudo-random group element in $\mathbb{G}$ by transforming a random string of length $\lambda$.
        \item $\mathtt{eSpat}.\mathtt{KeyGen}(\lambda, \alpha, \beta)$ $\to$ $(k_0, k_1)$: With a security parameter denoted as $\lambda$, a string represented by $\alpha$, and values $\beta$, this algorithm provides keys $(k_0,k_1)$.
        \item $\mathtt{eSpat}.\mathtt{Eval}(b,k_b,x)\to y_b$: This algorithm takes as input server $b$, along with corresponding key $k_b$, and string $x$. It generates a value $y_b$.
    \end{itemize}
 
    Since spatial data is constantly updating, clients need to upload updated data in real-time. Generally, clients need to upload a key to offset the old data information, and then re-upload the new key for the updated data.
\end{definition}

\begin{figure*}
    \centering
    \includegraphics[width=0.85\textwidth]{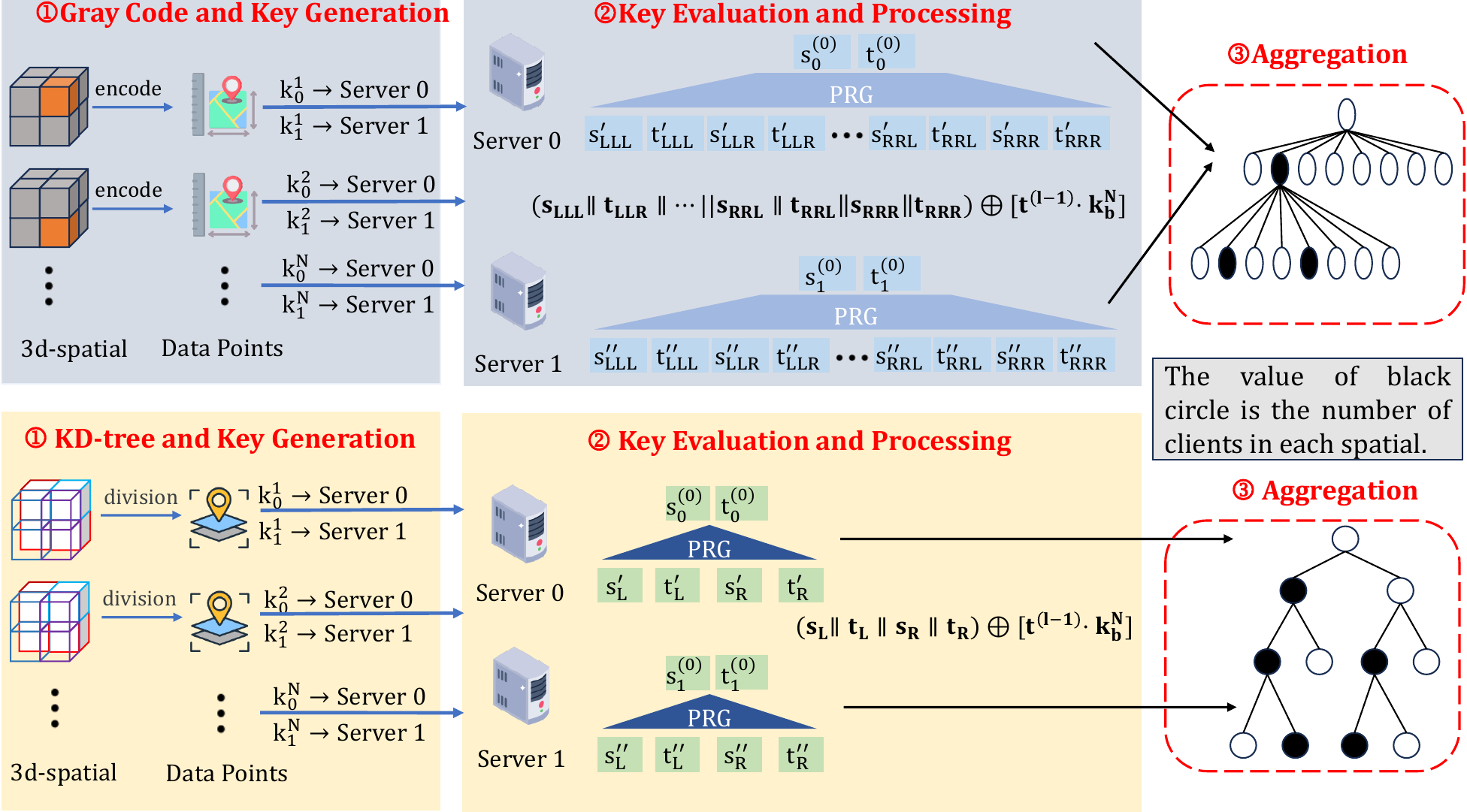}
    \caption{\textcolor{black}{Workflow of $\mathsf{eSpat}\mbox{-}B$ and $\mathsf{eSpat+}$. 
In $\mathsf{eSpat}\mbox{-}B$ shown in the upper part, step \ding{172} maps spatial data to Gray-code-based octree indices and generates DPF key shares. 
In $\mathsf{eSpat+}$ shown in the lower part, step \ding{172} performs KD-tree-based spatial division and generates prefix-oriented incremental DPF key shares. 
Steps \ding{173} and \ding{174} are similar in both schemes: the two non-colluding servers evaluate the received key shares, and the requester aggregates the returned partial results to obtain the final statistics.}}
    \label{fig: workflow}
    \vspace{-.15in}
\end{figure*}

\subsection{Design Goals}
The main design goal is to achieve efficient spatial distribution statistics analysis while protecting the privacy of spatial data. Specifically, the design goals of the proposed scheme are summarized as follows.

\begin{itemize}
    \item {\textbf{Confidentiality.}} The scheme ensures the privacy of spatial data while preventing specific information from being inferred from statistical results.
    \item {\textbf{Functionality.}} The scheme supports accurate spatial data distribution statistics analysis, allowing clients to share spatial data in a private manner.
    \item {\textbf{Efficiency.}} We design the scheme to minimize computational and communication overhead, enabling timely spatial distribution statistics analysis.
\end{itemize}

\section{Privacy-Preserving Distribution Statistics}\label{sec:4}
In this section, we first propose a scheme called $\mathsf{eSpat\mbox{-}B}$ based on the Gray code and octree. To improve performance, we design another spatial distribution statistics analysis scheme called $\mathsf{eSpat+}$ based on incremental DPF, which is introduced in detail below. The workflow is illustrated in Fig.~\ref{fig: workflow}.

\subsection{$\mathsf{eSpat\mbox{-}B}$: Basic $\mathsf{eSpat}$ Scheme}
\subsubsection{Main Idea}

\textcolor{black}{Each client uploads key shares, generated through a key generation algorithm, to two servers, which aggregate the shares for distribution statistics analysis on spatial data. If standard DPF is used to perform spatial distribution statistics, it incurs high computational overhead due to its binary tree structure. Since spatial data distribution statistics analysis requires tripling the tree height, the efficiency of DPF is heavily constrained by this factor. To overcome this limitation, we improve DPF by integrating Gray code for spatial encoding and utilizing an octree structure for indexing sub-regions, which reduces computational complexity and enhances efficiency for three-dimensional spatial distribution statistics analysis.}

\textcolor{black}{To construct our scheme, each spatial data point is encoded as $g_U \gets Gray(U)$ ($U$ is the spatial data point), and the statistical area is divided accordingly. Each parent area is divided into $8$ subareas (with $3$ bits defining a spatial area), ensuring that the parent code is a common prefix of the subarea codes. The client then generates two key shares via $\mathtt{eSpat}.\mathtt{KeyGen}$ and sends them to the servers, which use $\mathtt{eSpat}.\mathtt{Eval}$ to aggregate key shares. The requester computes the final statistical result.}

\subsubsection{Scheme Details}
$\mathsf{eSpat\mbox{-}B}$ mainly consists of three steps: key generation, data processing, and data statistics. Each client possesses a spatial data encoding $\alpha$ that is represented in Gray code. We will now elaborate on the specifics of the distribution statistics analysis process.

% \begin{figure*}
%     \centering
%     \includegraphics[width=0.8\textwidth]{Figure/Workflow of $\mathsf{eSpat\mbox{-}B}$.pdf}
%     \caption{Workflow of $\mathsf{eSpat\mbox{-}B}$.}
%     \label{fig: workflow}
% \end{figure*}

\textbf{Setup}: This step establishes a pseudo-random generator $G$ : $\{0,1\}^\lambda \rightarrow \{0,1\}^{8\lambda+8}$ and a pseudo-random group element in $\mathbb{G}$ that transforms a random string of length $\lambda$ denoted as $\text{Convert}_{\mathbb{G}^{\prime}}$: $\{0,1\}^\lambda \rightarrow \mathbb{G}^{\prime}$, where $\mathbb{G^{\prime}}:=\{0,1\}^{\lambda}\times \mathbb{G}$.

\begin{algorithm}[!ht]    \SetAlgoLined %显示end
	\caption{$\mathtt{eSpat}.\mathtt{KeyGen}$}%算法名字
        \label{DPF:Alg1}
	\KwIn{$1^{\lambda}, \alpha, \beta$}%输入参数
	\KwOut{$k_0, k_1$}%输出
    Set $\alpha = \alpha_1,\dots,\alpha_n\in \{0,1\}^n$ be the bit decomposition, where $\alpha_i \in (000,001,\ldots)$\;
	Select randomly $s_0^{(0)}\gets \{0,1\}^{\lambda}$ and $s_1^{(0)}\gets \{0,1\}^{\lambda}$\; %\;用于换行
        Let $t_0^{(0)} = 0$ and $t_1^{(0)} = 1$\;
	\For{$l = 1 $ to $n$}{
		$s_b^{\mathsf{LLL}} || t_b^{\mathsf{LLL}} || s_b^{\mathsf{LLR}} || t_b^{\mathsf{LLR}} || \ldots|| s_b^{\mathsf{RLL}} || t_b^{\mathsf{RLL}} || s_b^{\mathsf{RRR}}$ \\
  $|| t_b^{\mathsf{RRR}} \gets G(s_b^{(l-1)})$ for $b = 0,1$\;
		\If{$\alpha_l = 000$}{
			$\mathsf{Keep}\gets \mathsf{LLL}$, $\mathsf{Lose1}\gets \mathsf{LLR}$,$\ldots$,\\
                $\mathsf{Lose4}\gets \mathsf{RLL}$,$\ldots$,$\mathsf{Lose7}\gets \mathsf{RRR}$\;
		}
            \If{$\alpha_l = 001$}{
			$\mathsf{Keep}\gets \mathsf{LLR}$, $\mathsf{Lose1}\gets \mathsf{LLL}$,$\ldots$,\\
                $\mathsf{Lose4}\gets \mathsf{RLL}$, $\ldots$,$\mathsf{Lose7}\gets \mathsf{RRR}$\; 
		}
            $\ldots$\;
            \If{$\alpha_l = 101$}{
			$\mathsf{Keep}\gets \mathsf{RRL}$, $\mathsf{Lose1}\gets \mathsf{LLL}$,$\ldots$,\\ 
                $\mathsf{Lose4}\gets \mathsf{RLL}$, $\ldots$, $\mathsf{Lose7}\gets \mathsf{RRR}$\;
		}
            \If{$\alpha_l = 100$}{
			$\mathsf{Keep}\gets \mathsf{RRR}$, $\mathsf{Lose1}\gets \mathsf{LLL}$,$\ldots$,\\
                $\mathsf{Lose4}\gets \mathsf{RLL}$, $\ldots$,$\mathsf{Lose7}\gets \mathsf{RRL}$\;
		}
            $s_{CW}^{\mathsf{Keep}} \gets 0$\;
           \For{$i=1$ to $7$}{
           $s_{CW[i]}\gets s_0^{\mathsf{Lose[i]}} \oplus s_1^{\mathsf{Lose[i]}}$,\\
           $t_{CW[i]}\gets t_0^{\mathsf{Lose[i]}} \oplus t_1^{\mathsf{Lose[i]}}$,\\
           $s_{CW}^{\mathsf{Keep}}\gets s_{CW}^{\mathsf{Keep}}\oplus s_{CW[i]}$,\\
           
           }
            $t_{CW}^{\mathsf{Keep}}\gets t_0^{\mathsf{Keep}}\oplus t_1^{\mathsf{Keep}}\oplus 1$\;
            $t_b^{(l)} \leftarrow t_b^{\mathsf{Keep}} \oplus t_b^{(l-1)} \cdot t_{CW}^{\mathsf{Keep}}$ for $b =0,1$\;
            $s_b^{(l)} \leftarrow s_b^{\mathsf{Keep}} \oplus t_b^{(l-1)} \cdot s_{CW}^{\mathsf{Keep}}$ for $b =0,1$\;
            $CW^{(l)} \leftarrow s_{CW}^{\mathsf{Keep}} || t_{CW}^{\mathsf{Keep}} || s_{CW[1]} || t_{CW[1]} $ $|| s_{CW[2]} || 
            t_{CW[2]} ||\ldots||s_{CW[7]}||t_{CW[7]}$\;
	}
       $CW^{(n+1)}\gets (-1)^{t_1^n}[\beta- \text{Convert}(s_0^{(n)})+\text{Convert}(s_1^{(n)})]$\;
	Let $k_b\gets s_b^{(0)}||t_b^{(0)}||CW^{(1)}||\ldots||CW^{(n+1)}$ for $b= 0,1$\;
	\Return $(k_0, k_1)$
\end{algorithm}

\textbf{Key generation.} In this step, each client generates keys to safeguard the privacy of spatial data. Specifically, each client derives a spatial data code $\alpha$ using Gray code and generates key shares through $\mathtt{eSpat}.\mathtt{KeyGen}$. These key shares are then transmitted to the cloud servers.

$\mathtt{eSpat}.\mathtt{KeyGen}(1^{\lambda}, \alpha , \beta)\to (k_0,k_1)$: The client decomposes the data encoding $\alpha$ into a bit sequence $(\alpha_1, \alpha_2, \ldots, \alpha_n)$. For each bit $\alpha_i$, its value (e.g., $000, 001, 010$) determines a $\mathsf{Keep}$ path for the correct path and generates $\mathsf{Lose}$ paths for incorrect ones (Alg. \ref{DPF:Alg1} Lines 7-23), where $\mathsf{LLL},\mathsf{LLR},\ldots,\mathsf{RRL},\mathsf{RRR}$ denote eight directions in the octree. The client starts by generating random seeds $s_0^{(0)}$ and $s_1^{(0)}$ for two servers and sets control bits $t_0^{(0)} = 0$ and $t_1^{(0)} = 1$. During the main loop, the client processes each bit $\alpha_i$, applying XOR operations to combine selected $\mathsf{Keep}$ and $\mathsf{Lose}$ paths. This path information is accumulated in $s_{CW[i]}$ and $t_{CW[i]}$ (Alg. \ref{DPF:Alg1} Lines 25-30). Next, once the entire bit sequence is processed, the accumulated path data is used to generate final key shares $k_0$ and $k_1$, where $CW^{(l)}$ is the correction words of $l_{th}$ layer. The detailed algorithm is shown in Alg. \ref{DPF:Alg1}.

\textbf{Data processing.} The server $b$ receives key shares from each client and utilize $\mathtt{eSpat}.\mathtt{Eval}$ to aggregate them. Next, each server sends the aggregated results to the requester.

$\mathtt{eSpat}.\mathtt{Eval}(b,k_b,x)\to y_b$: This algorithm is conducted by servers and is responsible for computing and aggregating the key shares $k_b$ from each client to obtain the partial result of the statistical spatial $x$. Each server first decomposes $k_b$ into the initial seed $s^{(0)}$, control bit $t^{(0)}$, and correction words $CW^{(l)}$ for each level. Then, each server iteratively processes each bit $x_l$, selecting the $\mathsf{Keep}$ path (representing the correct path) or $\mathsf{Lose}$ paths (representing incorrect paths) based on the bit value, thereby updating the data information $\tau^{(l)}$. At each level, each server parses the path fragments $s^{\mathsf{LLL}}, s^{\mathsf{LLR}}, \dots$ according to the current bit $x_l$ and updates the values of $s^{(l)}$ and $t^{(l)}$. After completing all levels, the algorithm uses the final seed $s^{(n)}$, control bit $t^{(n)}$, and final path information $CW^{(n+1)}$ to compute the output $y_b$ (Alg.\ref{DPF:Alg2} Line 18). The resulting $y_b$ is the aggregated partial statistical result based on the key shares. Next, servers send $y_b$ to the requester. The specific algorithm is presented in Alg. \ref{DPF:Alg2}.

\begin{algorithm}[]    \SetAlgoLined %显示end
	\caption{$\mathsf{eSpat}.\mathsf{Eval}$}%算法名字
        \label{DPF:Alg2}
	\KwIn{$b,k_b,x$}%输入参数
	\KwOut{$y_b$}%输出
    Parse $k_b=s^{(0)}||t^{(0)}||CW^{(1)}||\ldots||CW^{(n+1)}$\;
	\For{$l = 1 $ to $n$}{
        Parse $CW^{(l)} = s_{CW}^{\mathsf{Keep}} || t_{CW}^{\mathsf{Keep}} || s_{CW[1]} || t_{CW[1]} $ $|| s_{CW[2]} || t_{CW[2]} ||\ldots||s_{CW[7]}||t_{CW[7]}$\;
        $\tau^{(l)}\gets G(s^{(l-1)})\oplus (t^{(l-1)}\cdot CW^{(l)})$\;
        Parse $\tau^{(l)}=s^{\mathsf{LLL}}||t^{\mathsf{LLL}}||s^{\mathsf{LLR}}||t^{\mathsf{LLR}}||\ldots||s^{\mathsf{RRL}}||t^{\mathsf{RRL}}||s^{\mathsf{RRR}}$\\$|| t^{\mathsf{RRR}}\in \{0,1\}^{8\lambda+8}$\;
        \If{$x_l=000$}{
            $s^{(l)}\gets s^{\mathsf{LLL}}$,$t^{(l)}\gets t^{\mathsf{LLL}}$\;
        }
        \If{$x_l=001$}{
            $s^{(l)}\gets s^{\mathsf{LLR}}$, $t^{(l)}\gets t^{\mathsf{LLR}}$\;
        }
        $\ldots$\;
        \If{$x_l=100$}{
            $s^{(l)}\gets s^{\mathsf{RRR}}$, $t^{(l)}\gets t^{\mathsf{RRR}}$\;
        }
 }
        $y_b=(-1)^b[\text{Convert}(s^{(n)})+t^{(n)}\cdot CW^{(n+1)}]\in \mathbb{G}$\;
	\Return $y_b$
\end{algorithm}

\textbf{Data statistics and update.} After receiving the aggregated results, the requester can obtain the distribution of spatial data as follows:
\begin{equation}
\mathsf{eSpat}.\mathsf{Eval}(k_0, x) + \mathsf{eSpat}.\mathsf{Eval}(k_1, x) = 
\begin{cases}
1, & \text{$x = \alpha$}\\
0, & \text{$x \neq \alpha$}
\end{cases}.
\end{equation}

When the spatial data changes, it updates the data by generating two key shares. First, the client generates key shares to offset the old data and sends them to the server to remove its influence from the statistics of the original data. Then, the client generates new key shares representing the updated data and sends these key shares to the server for aggregation. Specifically, the client generates the offset key share through the $\mathtt{eSpat}.\mathtt{KeyGen}(1^{\lambda},\alpha_{old},-1)$, and then calculates the new key share through $\mathtt{eSpat}.\mathtt{KeyGen}(1^{\lambda},\alpha_{new},1)$ again.

\subsection{$\mathsf{eSpat+}$: Incremental DPF-Based Scheme}
\subsubsection{Differences and Improvements}
Firstly, compared to the basic scheme, $\mathsf{eSpat+}$ encodes the statistical region using a KD-tree, a multidimensional spatial segmentation structure that divides space into sub-regions. This approach reduces redundancy in region segmentation. By improving spatial indexing efficiency, KD-tree encoding enables faster positioning and statistics processing.

Secondly, $\mathsf{eSpat+}$ utilizes incremental DPF to implement prefix statistics. While the standard DPF in the basic scheme handles precise data statistics, the incremental DPF enables prefix matching, allowing for the simultaneous processing of multiple sub-regions. This enhancement facilitates the direct calculation of statistical results for multiple areas during range statistics, significantly boosting efficiency.

Finally, $\mathsf{eSpat+}$ incorporates an efficient update mechanism to address the frequent position changes of spatial data. Unlike the basic scheme, which requires a complete update of key shares for every position change, $\mathsf{eSpat+}$ updates only the key share for the affected area once. Since spatial data is often confined to small shifts between sub-areas, focusing on these movements during updates improves efficiency by approximately $50\%$, making the system more suitable for real-time spatial applications.

\textcolor{black}{Compared with $\mathsf{eSpat\mbox{-}B}$, $\mathsf{eSpat+}$ improves update efficiency by exploiting prefix sharing in the KD-tree partition. In $\mathsf{eSpat\mbox{-}B}$, each updated spatial data is encoded as a complete spatial index, and the corresponding DPF keys need to be regenerated and evaluated for that complete index. Therefore, updates are handled at the granularity of individual encoded points or regions. In contrast, $\mathsf{eSpat+}$ organizes the spatial domain as a KD-tree, where each node represents a spatial subregion and the path from the root to the node forms a prefix. The incremental DPF operates on these prefixes. When a spatial point is inserted, deleted, or modified, only the prefixes associated with the affected KD-tree subregions need to be updated. Unaffected subregions keep their previous states and do not need to be recomputed.}

\textcolor{black}{Moreover, if multiple updated records fall under the same parent region, they share the same prefix in the KD-tree. In this case, $\mathsf{eSpat+}$ can aggregate these updates at the shared prefix level, avoiding repeated DPF generation and evaluation for each individual record. This prefix-based update mechanism reduces both computation and communication overhead, which explains the better update efficiency of $\mathsf{eSpat+}$.}

\subsubsection{Detailed Scheme Construction} 

\textcolor{black}{Similar to $\mathsf{eSpat}\mbox{-}B$, $\mathsf{eSpat+}$ also mainly includes key generation, data processing, and statistics. In the construction of the $\mathsf{eSpat+}$ solution, the core is to adopt a static and publicly available KD-tree as the spatial index structure. The dividing points of all internal nodes of this tree are predetermined and made public during the system initialization. These dividing points are selected based on the median to ensure the generation of a balanced depth and evenly partitioned index tree. The specific choice of split points within the KD-tree does not affect the cryptographic correctness or privacy guarantees of the computation. Instead, these points define the spatial granularity and boundaries of the statistical query itself, which are predetermined by the system or specified by the requester based on the analysis needs (e.g., dividing a city into administrative districts or uniform grids). Each internal node splits a spatial region along one dimension, and each branch corresponds to one partitioning decision. Therefore, the path from the root to any node defines a prefix, where each bit records the left/right branching decision at one KD-tree split. Each prefix represents the spatial subregion associated with that node. Incremental DPF uses these prefixes to update only the affected subregions, thereby avoiding recomputation over the entire spatial domain.}

% \textcolor{black}{In $\mathsf{eSpat+}$, the spatial domain is recursively partitioned using a KD-tree. Each internal node splits a spatial region along one dimension, and each branch corresponds to one partitioning decision. Therefore, the path from the root to any node defines a prefix, where each bit records the left/right branching decision at one KD-tree split. Each prefix represents the spatial subregion associated with that node. Incremental DPF uses these prefixes to update only the affected subregions, thereby avoiding recomputation over the entire spatial domain.}

\textbf{Setup}. In Alg. \ref{Alg:kd-tree}, a 3-dimensional tree is constructed by recursively dividing the point set to create a hierarchical structure. The algorithm begins by determining the splitting axis based on the current depth, cycling through the $x$, $y$, and $z$ axes. It then sorts the points along the current axis and selects the median point as the node. The point set is split into left and right subsets, and the algorithm recursively builds the left and right subtrees, incrementing the depth at each level to switch the splitting axis. Recursion stops when no points remain, resulting in a balanced 3D tree. Additionally, similar to the \textbf{Setup} in $\mathsf{eSpat\mbox{-}B}$, a pseudo-random generator $G$ is established: $G: \{0,1\}^\lambda \to \{0,1\}^{2\lambda+2}$, and $\text{Convert}_{\mathbb{G}^{\prime}} \to \mathbb{G}^{\prime}$.

\begin{algorithm}[]
\caption{3D-Tree Construction}
\label{Alg:kd-tree}
\KwIn{A set of points in 3D space, $depth = 0$}
\KwOut{Root of the constructed 3D KD-Tree}
\SetKwFunction{FMain}{BuildKDTree}
\SetKwProg{Fn}{Function}{:}{}
\Fn{\FMain{points, depth}}{
    \If{points is empty}{
        \Return NULL\;
    }
    % Determine the axis based on depth
    $axis = depth \mod 3$\;
    % Sort points by the current axis and select the median as pivot
    Sort points by the current axis\;
    $median = \text{len(points)} // 2$\;
    % Create node and construct subtrees
    $node = new Node(points[median])$\;
    $node.left = \FMain{points[:median], depth + 1}$\;
    $node.right = \FMain{points[median + 1:], depth + 1}$\;
    \Return node\;
}
\end{algorithm}

\textbf{Key generation.} In Alg. \ref{IDPF:Alg1}, the client initializes random seeds $s_0^{(0)}$ and $s_1^{(0)}$, sets control bits $t_0^{(0)}$ and $t_1^{(0)}$, and defines the initial depth as $0$. For each level $l$, the client determines the splitting axis (cycling through $x$, $y$, and $z$ axes), selects a $\mathsf{Keep}$ path based on its coordinate, and accumulates the path information for both $\mathsf{Keep}$ and $\mathsf{Lose}$ paths. Similar to $\mathsf{eSpat\mbox{-}B}$, this process updates strings $s_b^{(l)}$ and $t_b^{(l)}$. After all levels are processed, the accumulated information is used to generate final key shares.

\begin{algorithm}[!ht]    \SetAlgoLined %显示end
	\caption{$\mathtt{eSpat+}.\mathtt{KeyGen}$}%算法名字
        \label{IDPF:Alg1}
	\KwIn{$1^{\lambda}, ((x_0,y_0,z_0), (\mathbb{G}_1, \beta_1),\ldots, (\mathbb{G}_n,\beta_n))$}%输入参数
	\KwOut{$k_0, k_1, pp$}%输出
	 $s_0^{(0)}\gets \{0,1\}^{\lambda}$ and $s_1^{(0)}\gets \{0,1\}^{\lambda}$\; %\;用于换行
        $t_0^{(0)} = 0$ and $t_1^{(0)} = 1$, $depth=0$\;
	\For{$l = 1 $ to $n$}{
		$s_b^{\mathsf{L}} || t_b^{\mathsf{L}} || s_b^{\mathsf{R}} || t_b^{\mathsf{R}} \gets G(s_b^{(l-1)})$ for $b = 0,1$\;
		$axis = depth \mod 3$\;  
    \If{$axis = 0$}{ % Splitting on x-axis
    \eIf{$x_0 < \text{node.x}$}{
        $\mathsf{Keep} \leftarrow Left$,
        $\mathsf{Lose} \leftarrow Right$\;
    }{
        $\mathsf{Keep} \leftarrow Right$,
        $\mathsf{Lose} \leftarrow Left$\;
    }
}
\ElseIf{$axis = 1$}{ % Splitting on y-axis
    \eIf{$y_0 < \text{node.y}$}{
        $\mathsf{Keep} \leftarrow Left$,
        $\mathsf{Lose} \leftarrow Right$\;
    }{
        $\mathsf{Keep} \leftarrow Right$,
        $\mathsf{Lose} \leftarrow Left$\;
    }
}
\Else{ % Splitting on z-axis
    \eIf{$z_0 < \text{node.z}$}{
        $\mathsf{Keep} \leftarrow Left$,
        $\mathsf{Lose} \leftarrow Right$\;
    }{
        $\mathsf{Keep} \leftarrow Right$,
        $\mathsf{Lose} \leftarrow Left$\;
    }
}
            $s_{CW} \leftarrow s_0^{\mathsf{Lose}} \oplus s_1^{\mathsf{Lose}}$\;
            $t_{CW}^\mathsf{L} \leftarrow t_0^{\mathsf{L}} \oplus t_1^{\mathsf{L}}\oplus \alpha_l \oplus 1, t_{CW}^\mathsf{R}\leftarrow t_0^{\mathsf{R}} \oplus t_1^{\mathsf{R}}$\;
            $t_b^{(l)} \leftarrow t_b^{\mathsf{Keep}} \oplus t_b^{(l-1)} \cdot t_{CW}^{\mathsf{Keep}}$ for $b =0,1$\;
            $\tilde{s}_b^{(l)} \leftarrow s_b^{\mathsf{Keep}} \oplus t_b^{(l-1)} \cdot s_{CW}$ for $b =0,1$\;
            $s_b^{(l)} || W_b^{(l)} \leftarrow \text{Convert}_{\mathbb{G}_{l}}(\tilde{s}_b^{(l)})$ for $b = 0,1$\;
            $W_{CW}^{(l)} \leftarrow (-1)^{t_1^{(l)}} \cdot [\beta_{l} - W_{0}^{(l)} + W_{1}^{(l)}]$\;
            $CW^{(l)} \leftarrow s_{CW} || t_{CW}^{\mathsf{L}} || t_{CW}^\mathsf{R} || W_{CW}^{(l)}$\;
            $depth++$\;
	}
        $s_0^{\prime} = s_0^{(n)}$, $s_1^{\prime} = s_1^{(n)}$, $t_0^{\prime} = t_0^{(n)}$, $t_1^{\prime} = t_1^{(n)}$\;
	Let $k_b\gets s_b^{(0)}$ for $b= 0,1$\;
        Let $pp \gets CW^{(1)},\ldots, CW^{(n)}$\;
	\Return $(k_0, k_1 ,pp)$, $(s_0^{\prime},s_1^{\prime},t_0^{\prime},t_1^{\prime})$
\end{algorithm}

Then, the client generates new key shares when spatial data is updated. It initializes random seeds $s_0^{(0)}$ and $s_1^{(0)}$, and control bits $t_0^{(0)} = 0$ and $t_1^{(0)} = 1$. The client calls $\mathtt{eSpat+}.\mathtt{KeyGen}$ three times: first for the public region $R(x_0, y_0, z_0)$, then for the old range $R(x', y', z')$ to remove its influence, and finally for the updated range $R(x'', y'', z'')$ to include the new data. The client then combines the key shares and correction words into $k_0$, $k_1$, and $pp$, $pp_{old}$, and $pp_{new}$. Details refer to Alg. \ref{IDPF:Alg2}. 

\textbf{Data processing.} Each server interprets the input key share and the correction word set $pp$ to statistically analyze the distribution of the range $R(x), R(y), R(z)$, where $R(\cdot)$ denotes $\cdot>node.value$. The server initializes the depth and processes the seed and $CW^{(l)}$. It then selects the axis based on the current depth, and compares the node's coordinates with the specified range. Depending on whether the node coordinates fall within the range, the server chooses the left ($\mathsf{L}$) or right ($\mathsf{R}$) path, updating the state with $s^\mathsf{L}$, $t^\mathsf{L}$ or $s^\mathsf{R}$, $t^\mathsf{R}$. After processing all levels, the server returns the final result $y_b^l$ and $st^l$. The detailed algorithm is shown in Alg. \ref{IDPF:Alg3}.

\begin{algorithm}[]
    \SetAlgoLined %显示end
	\caption{$\mathtt{eSpat+}.\mathtt{MoveGen}$}%算法名字
        \label{IDPF:Alg2}
	\KwIn{$1^\lambda$,(($R(x_0,y_0,z_0)$,($\mathbb{G}_1$,$\beta_1$),$\dots$,($\mathbb{G}_m$,$\beta_m$)),
        \\ ($R(x',y',z')$,($\mathbb{G}_{m+1}$,$\beta_{m+1}$),$\dots$,($\mathbb{G}_{n}$,$\beta_{n}$)),
        \\ ($R(x'',y'',z'')$,($\mathbb{G}_{m+1}$,$\beta_{m+1}'$),$\dots$,($\mathbb{G}_{n}$,$\beta_{n}'$)))}%输入参数
	\KwOut{$k_0,k_1,pp,pp_{old},pp_{new}$}%输出
       
     $s_0^{(0)}\stackrel{R}{\longleftarrow}\{0,1\}^\lambda$ and $s_1^{(0)}\stackrel{R}{\longleftarrow}\{0,1\}^\lambda$\;
	Let $t_0^{(0)}\leftarrow 0$ and $t_1^{(0)}\leftarrow 1$\;
        $(s_0,s_1,t_0,t_1,pp) \leftarrow$ $\mathtt{eSpat+}.\mathtt{KeyGen}$
        $(1^\lambda,(R(x_0,y_0,z_0),(\mathbb{G}_1,\beta_1),\ldots,(\mathbb{G}_m,\beta_m)))$\;
	$(s_0,s_1,t_0,t_1,pp_{old})        
        \leftarrow$ $\mathtt{eSpat+}.\mathtt{KeyGen}$ 
        $(1^\lambda,(R(x',y',z'),(\mathbb{G}_{m+1},\beta_{m+1}),\dots,(\mathbb{G}_{n},\beta_{n})));$     
        $(s_0,s_1,t_0,t_1,pp_{new}) \leftarrow$ $\mathtt{eSpat+}.\mathtt{KeyGen}$ $(1^\lambda, (R(x'',y'',z''),(\mathbb{G}_{m+1},\beta_{m+1}'),\dots,(\mathbb{G}_{n},\beta_{n}')))$\;
        Let $k_b \leftarrow s_b^{(0)}$ for $b = 0,1$\;
	\Return ($k_0,k_1,pp,pp_{old},pp_{new}$)
\end{algorithm}

Next, each server performs prefix statistics by calling $\mathtt{eSpat+}. \mathtt{EvalNext}$ at each level, using the previous key share and the current $CW^{(j)}$ along with the range. This produces updated results $(st_b^{j}, y_b^{j})$ at each level (Alg. \ref{IDPF:Alg4} Line 5). After iterating through all levels, the server returns the accumulated result $y_b^l$, which represents the partial statistical result.

In Alg. \ref{IDPF:Alg5}, when the spatial data is updated, the server evaluates the new data and ensures the correctness of the aggregation results. The server initializes $s^{(0)} = k_b$ and $t^{(0)} = b$, and interprets the common correction words set $pp$, along with old and new  information $pp_{old}$ and $pp_{new}$. The server first processes the common data by calling $\mathtt{eSpat+}.\mathtt{EvalNext}$ for each layer, then separately evaluates the old and new data, adjusting the result by subtracting the impact of old data. The final updated result $y_{b, new}$ is returned.

\textbf{Data statistics.} Similar to $\mathsf{eSpat\mbox{-}B}$, the requester can obtain the spatial distribution of spatial data by directly adding the partial results $y_b^l$ from the two servers.

\begin{algorithm}[]
    \SetAlgoLined %显示end
	\caption{$\mathtt{eSpat+}.\mathtt{EvalNext}$}%算法名字
        \label{IDPF:Alg3}
	\KwIn{$b,st_b^{l-1},pp_{l},R(x),R(y),R(z)$}%输入参数
	\KwOut{$st^{l},y_b^{l}$}%输出
	Interpret $st_b^{l-1}$ = $s^{(l-1)} || t^{(l-1)}$\; %\;用于换行
        % $depth = 0$\;
        Interpret $G(s^{(l-1)}) = \hat{s}^{\mathsf{L}} || \hat{t}^{\mathsf{L}} || \hat{s}^{\mathsf{R}} || \hat{t}^{\mathsf{R}}$\;
        Interpret $CW^{(l)}=s_{CW} || t^{\mathsf{L}}_{CW} || t^{\mathsf{R}}_{CW} || W_{CW}$;\\
        $\tau^{l}\leftarrow$ $(\hat{s}^{\mathsf{L}}||\hat{t}^{\mathsf{L}}||\hat{s}^{\mathsf{R}}||\hat{t}^{\mathsf{R}})\oplus(t^{(l-1)}\cdot [s_{CW}||t^{\mathsf{L}}_{CW}||s_{CW}||t^{\mathsf{R}}_{CW}])$\;
        Interpret $\tau^{l} =s^{\mathsf{L}} || t^{\mathsf{L}} || s^{\mathsf{R}} || {t}^{\mathsf{R}} \in \{0,1\}^{2\lambda+2}$\;
	$axis = l \mod 3$\;  
    \If{$axis = 0$}{ % Splitting on x-axis
    \eIf{$R(x) == (x < node.x)$ }{
        $\tilde{s}^{(l)}\gets s^\mathsf{L}, t^{(l)}\gets t^\mathsf{L}$\;
    }{
        $\tilde{s}^{(l)}\gets s^\mathsf{R}, t^{(l)}\gets t^\mathsf{R}$\;
    }}
    \ElseIf{$axis = 1$}{ % Splitting on y-axis
        \eIf{$R(y) == (y < node.y) $}{
        $\tilde{s}^{(l)}\gets s^\mathsf{L}, t^{(l)}\gets t^\mathsf{L}$\;
    }{
        $\tilde{s}^{(l)}\gets s^\mathsf{R}, t^{(l)}\gets t^\mathsf{R}$\;
    }
}
    \Else{ % Splitting on z-axis
        \eIf{$R(z) == (z < node.z) $}{
        $\tilde{s}^{(l)}\gets s^\mathsf{L}, t^{(l)}\gets t^\mathsf{L}$\;
    }{
        $\tilde{s}^{(l)}\gets s^\mathsf{R}, t^{(l)}\gets t^\mathsf{R}$\;
    }
}
       $s^{(l)} || W^{(l)} \leftarrow Convert_{\mathbb{G}_{l}}(\tilde{s}^{(l)})$\;
       $st^{l} \leftarrow s^{(l)} || t^{(l)}$\;
       $y_b^{l} \leftarrow (-1)^b \cdot [W^{(l)} + t^{(l)} \cdot W_{CW}]$\;
       \Return ($st^{l},y_b^{l}$)
\end{algorithm}

\section{Theoretical Analysis}\label{sec:5}
In this section, we theoretically analyze the security and complexity of $\mathsf{eSpat\mbox{-}B}$ and $\mathsf{eSpat+}$.

\subsection{Security Analysis}
In this part, we analyze the security of $\mathsf{eSpat\mbox{-}B}$ and $\mathsf{eSpat+}$. Then, we prove that both schemes can protect data privacy.

\begin{theorem}\label{the2}
Under the semi-honest and non-colluding two-server model, $\mathsf{eSpat\mbox{-}B}$ and $\mathsf{eSpat+}$ protect the privacy of users' spatial data, assuming the key indistinguishability of the underlying DPF and IDPF primitives.
\end{theorem}

\begin{IEEEproof}
Let $\mathsf{View}*{S_b}$ denote the view of server $S_b$, where $b\in{0,1}$. In each execution, $\mathsf{View}*{S_b}$ consists of the public parameters, the key shares received from users, and the local values computed by $S_b$. Since the two servers are non-colluding, each server observes only one share of each DPF/IDPF key.

For $\mathsf{eSpat\mbox{-}B}$, each spatial data is encoded as a spatial index, and the client generates two DPF keys corresponding to this index. By the key indistinguishability of DPF, the key share held by any single server is computationally indistinguishable from a key share generated for any other spatial index in the same domain. Hence, for any two spatial data, the view of a single server is computationally indistinguishable, except for public parameters and unavoidable aggregate outputs.

For $\mathsf{eSpat+}$, the spatial domain is represented by KD-tree prefixes, and users generate IDPF keys for the corresponding prefixes. By the key indistinguishability of IDPF, a single server cannot distinguish key shares generated for different prefixes. Therefore, the server cannot infer the corresponding KD-tree subregion or the user's spatial data from its own view.

More formally, if there exists a probabilistic polynomial-time adversary that distinguishes the server view generated from two different spatial data with non-negligible advantage, then one can construct a distinguisher against the underlying DPF/IDPF primitive by embedding the challenge key share into the simulated server view. This contradicts the key indistinguishability of DPF/IDPF. By a standard hybrid argument over all user-generated keys, replacing real keys with keys for alternative records remains computationally indistinguishable to any single server. Therefore, under the semi-honest and non-colluding assumption, $\mathsf{eSpat\mbox{-}B}$ and $\mathsf{eSpat+}$ preserve the privacy of users' spatial data.
\end{IEEEproof}

\subsection{Complexity Analysis}

% \begin{table}[ht]
% \caption{Complexity between eSpat and other DPF schemes}
% \begin{tabular}{ll|llll}
% \hline
% \multicolumn{2}{l|}{}                                                          & \cite{boyle2016function}& \cite{boneh2021lightweight}& $\mathsf{eSpat\mbox{-}B}$ & $\mathsf{eSpat+}$ \\ \hline
% \multicolumn{1}{l|}{\multirow{2}{*}{Key size}}   & \multicolumn{1}{l|}{Any $n$} & $\approx n^2\lambda/2$      & $\approx 3n\lambda$             & $\approx 3n\lambda$  & $\approx n\lambda$    \\ \cline{2-6} 
% \multicolumn{1}{l|}{}                            & \multicolumn{1}{l|}{$n=256$} & $543$ KB        & $19.5$ KB           & $19.5$ KB   & $6.2$ KB  \\ \hline
% \multicolumn{1}{l|}{\multirow{2}{*}{Eval times}} & \multicolumn{1}{l|}{Any $n$} & $\approx 4n^2$        & $\approx 4n$             & $\approx 2n$  & $\approx n/2$    \\ \cline{2-6} 
% \multicolumn{1}{l|}{}                            & \multicolumn{1}{l|}{$n=256$} & $262142$       & $1026$            & $512$     & $127$   \\
% \hline

% \end{tabular}
% \end{table}

As shown in Table~\ref{tab:complex}, we compare the computational and communication complexities of $\mathsf{eSpat}$ with representative schemes for privacy-preserving spatial statistics analysis. Existing approaches such as LiangTC~\cite{liang2023privacy} and LiangTIFS~\cite{liang2024efficient} rely on expensive cryptographic primitives, including modular exponentiation, bilinear pairings, and repeated encryption operations, resulting in relatively high computational and communication costs. In contrast, DPF-based constructions primarily depend on lightweight pseudorandom generator (PRG) evaluations and therefore offer significantly better scalability.

For the \textit{KeyGen} phase, both $\mathsf{eSpat\mbox{-}B}$ and $\mathsf{eSpat+}$ achieve linear complexity $O(n)T_{PRG}$, matching the efficiency of BonehS\&P~\cite{boneh2021lightweight} while substantially improving upon the quadratic complexity of the standard DPF~\cite{boyle2016function}. For the \textit{Eval} phase, $\mathsf{eSpat\mbox{-}B}$ requires $O(n)T_{PRG}$, whereas $\mathsf{eSpat+}$ further reduces the complexity to $O(l)T_{PRG}$ by leveraging prefix-based evaluation, where $l \leq n$ denotes the query prefix length.

A similar trend can be observed for communication and storage overhead. Both variants of $\mathsf{eSpat}$ require only linear communication during key generation, while $\mathsf{eSpat+}$ achieves constant overhead in the evaluation phase, matching BonehS\&P and outperforming existing spatial query schemes. These improvements stem from the integration of hierarchical spatial partitioning with DPF-based aggregation, which eliminates the need to process large numbers of intermediate query results and enables efficient aggregation directly over encrypted spatial data.

% It can be seen from the comparison results that in terms of key size, $\mathsf{eSpat+}$ has the best key size, followed by $\mathsf{eSpat\mbox{-}B}$, and \cite{boyle2016function} has the largest key size. This is because the key size of \cite{boyle2016function} is significantly larger due to the need to support the complete prefix tree structure. On the other hand, $\mathsf{eSpat+}$ and $\mathsf{eSpat\mbox{-}B}$ have significant advantages in evaluations number. This is because the number of evaluations is related to the height of the tree. Each layer requires PRG expansion for each node. However, $\mathsf{eSpat\mbox{-}B}$ and $\mathsf{eSpat+}$ have fewer layers, they require fewer evaluations and offer higher computational efficiency.

\begin{algorithm}[]
    \SetAlgoLined %显示end
	\caption{$\mathtt{eSpat+}.\mathtt{EvalPre}$}%算法名字
        \label{IDPF:Alg4}
	\KwIn{$b,k_b,pp,r(x,y,z)\in R(x,y,z)$}%输入参数
	\KwOut{$y_b^{l}$}%输出
	$s^{(0)} = k_b$ and $t^{(0)} = b$\; 
        Interpret $pp = CW^{(1)},\ldots,CW^{(n)}$\;
        $st_b^0 \leftarrow s^{(0)} || t^{(0)}$;\\
        \For{$j=1$ to $l$}{
            $(st_b^j,y_b^j) \leftarrow$           
            $\mathtt{eSpat+}.\mathtt{EvalNext}$($b,st_b^{j-1},CW^{(j)},r(x,y,z)$);}
        \Return $y_b^{l}$
\end{algorithm}
% \vspace{-.1in}
\begin{algorithm}[]
    \SetAlgoLined %显示end
    \setlength{\baselineskip}{0.5cm}
	\caption{$\mathtt{eSpat+}.\mathtt{MoveEval}$}%算法名字
        \label{IDPF:Alg5}
	\KwIn{$b,k_b,pp,pp_{old},pp_{new}$,$R(x_0,y_0,z_0)$,$R(x',y',z')$,\\ $R(x'',y'',z'')$}%输入参数
	\KwOut{$y_{b,new}^l$}%输出
	Set $s^{(0)} = k_b$ and $t^{(0)} = b$\; %\;用于换行
        Interpret $pp = CW^{(1)},\ldots,CW^{(m)}$ $pp_{old} = CW_{old}^{(m+1)},\ldots,CW_{old}^{(n)}$ $pp_{new} = CW_{new}^{(m+1)},\ldots,CW_{new}^{(n)}$\;
        $st_b^0 \leftarrow s^{(0)} || t^{(0)}$\;
        \For{$j=1$ to $m$}{
            $(st_b^j,y_b^j) \leftarrow$ $\mathtt{eSpat+}.\mathtt{EvalNext}$($b,st_b^{j-1},CW^{(j)},R(x_0,y_0,z_0)$);
        }
        $st_{b,old}^0=st_{b,new}^0=st_b^m$\;
        \For{$ l = m+1 $ to $n$}{
            $(st_{b,old}^l,y_{b,old}^l) \leftarrow \mathtt{eSpat+}.\mathtt{EvalNext}(b,st_{b,old}^{l-1},CW^{(l)}_{old},R(x',y',z')$);
            $(st_{b,new}^l,y_{b,new}^l) \leftarrow$ $\mathtt{eSpat+}.\mathtt{EvalNext}$($b,st_{b,new}^{l-1},CW^{(l)}_{new}$,$R(x'',y'',z'')$);
$y_{b,new}^l=y_{b,old}^{l}+y_{b,new}^{l}$\;
            \Return $y_{b,new}^l$      
        }      
\end{algorithm}

\section{Experimental Evaluation}\label{sec:6}
In this section, we present the detailed experimental setups. Then, we offer an overview of the experimental results.

\subsection{Experimental Settings}

\textbf{Configuration}. The $\mathsf{eSpat\mbox{-}B}$ and $\mathsf{eSpat+}$ are employed to tackle distribution statistical issues in certain applications. Similar to Ref.~\cite{boneh2021lightweight}, the security parameter is set to $128$ bits. The test platform's configuration includes an AMD Ryzen 9 7945HX processor, an NVIDIA GeForce GTX 4070 graphics card, and 64 GB of system memory, with Java version 11.0.4. 

% As for baselines, we select LiangTC~searchable encryption~\cite{liang2023privacy}, LiangTIFS~range query~\cite{liang2024efficient}, and BonehS\&P~private aggregation~\cite{boneh2021lightweight}. The above schemes are used to compare the computation and communication overhead. In addition, we use the following two baseline schemes to compare the accuracy: LDP1~\cite{hong2022collecting} and LDP2~\cite{cunningham2021real}. Note that these methods are not originally designed for spatial data, but we extend them to support the spatial distribution statistics analysis. 

\textbf{Baselines.} We compare our schemes with representative approaches from four categories of privacy-preserving data analytics. Specifically, LiangTC~\cite{liang2023privacy} represents searchable encryption, LiangTIFS~\cite{liang2024efficient} represents privacy-preserving range query, and BonehS\&P~\cite{boneh2021lightweight} represents function-secret-sharing-based private aggregation. In addition, we include two local differential privacy schemes, denoted as LDP1~\cite{hong2022collecting} and LDP2~\cite{cunningham2021real}. These baselines cover the major technical paradigms related to privacy-preserving spatial statistics, providing a comprehensive performance comparison.

\textbf{Dataset.} We evaluate the proposed schemes using two real-world trajectory datasets: Geolife and T-Drive. Geolife contains 17,621 trajectories collected from 182 users, with each record consisting of time-stamped latitude, longitude, and altitude information. T-Drive contains approximately 15 million GPS points collected from 10,357 taxis over one week. Following our spatial encoding method, the trajectory points are transformed into spatial indexes and securely uploaded to the servers for statistical analysis.

\begin{table*}[]
\centering
\caption{Complexity between $\mathsf{eSpat}$ and other baselines. Note that the complexities listed in the table omit some operations with very low time and storage overhead, which can be neglected.}
\label{tab:complex}
\renewcommand{\arraystretch}{1.3}
\begin{tabular}{ccccccc}
\hline
\multicolumn{7}{c}{Computation}                                                                                                                     \\ \hline
\multicolumn{1}{c|}{}        & LiangTC~\cite{liang2023privacy}           & LiangTIFS~\cite{liang2024efficient}                   &  DPF~\cite{boyle2016function}           & BonehS\&P~\cite{boneh2021lightweight}     & $\mathsf{eSpat}\mbox{-}B$         & $\mathsf{eSpat+}$        \\\hline
\multicolumn{1}{c|}{$KeyGen$}  & $O(|B|+N)T_{EXP}$  & $O(Nn+M)T_{PRF}+O(M)T_{AES}$ & $O(n^2)T_{PRG}$ & $O(n)T_{PRG}$ & $O(n)T_{PRG}$ & $O(n)T_{PRG}$ \\
\multicolumn{1}{c|}{$Eval$} & $O(MN)T_{pair}$   & $O(logN)T_{AES}$             & $O(n)T_{PRG}$   & $O(l)T_{PRG}$ & $O(n)T_{PRG}$   & $O(l)T_{PRG}$ \\ \hline
\multicolumn{7}{c}{Storage}                                                                                                  \\ \hline
\multicolumn{1}{c|}{}        & LiangTC~\cite{liang2023privacy}            & LiangTIFS~\cite{liang2024efficient}                     & DPF~\cite{boyle2016function}            & BonehS\&P~\cite{boneh2021lightweight}     & $\mathsf{eSpat}\mbox{-}B$          & $\mathsf{eSpat+}$        \\ \hline
\multicolumn{1}{c|}{$KeyGen$}  & $O(|B|+N)\lambda$ & $O(Nn+M)\lambda$             & $O(n^2)\lambda$ & $O(n)\lambda$ & $O(n)\lambda$ & $O(n)\lambda$ \\
\multicolumn{1}{c|}{$Eval$} & $O(|B|+M\log N)\lambda$ & $O(log N)\lambda$            & $O(n)\lambda$   & $O(1)\lambda$ & $O(n)\lambda$   & $O(1)\lambda$ \\ \hline
\end{tabular}
\end{table*}

\begin{figure*}[]
  \centering % subfloat之间不能有空行，否则会被识别为换行。
  % \vspace{-.2in}
    \subfloat[][\normalsize{KeyGen time w.r.t string length.}]{
	\begin{minipage}{.3\linewidth}
    \includegraphics[width=\linewidth]{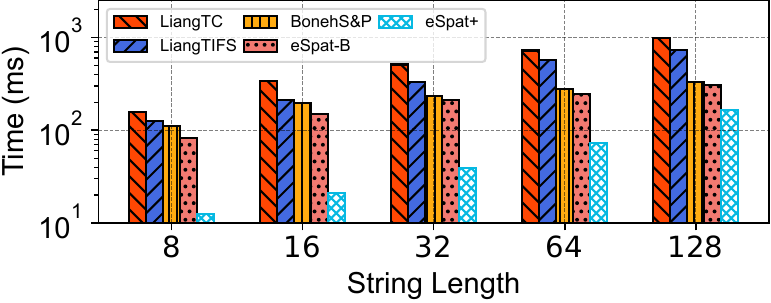}
    \vspace{-.2in}
    \label{fig:keygen-spatial}
    \end{minipage}
    }
    \vspace{-.15in}
    \subfloat[][\normalsize{KeyGen time w.r.t data records.}]{
	\begin{minipage}{.3\linewidth}
\includegraphics[width=\linewidth]{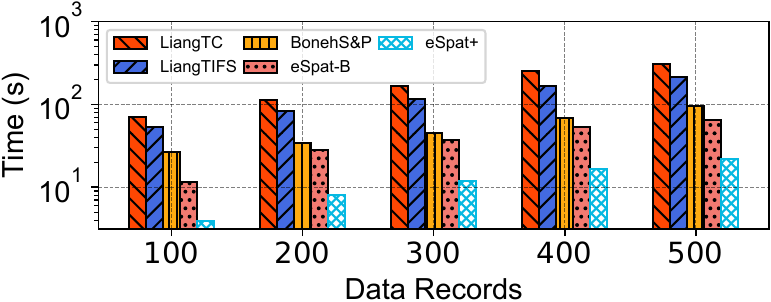}
    \vspace{-.2in}
    \label{fig:keygen-UAV}
    \end{minipage}
    } 
    \subfloat[][\normalsize{KeyEval time w.r.t string length.}]{
	\begin{minipage}{.3\linewidth}
    \includegraphics[width=\linewidth]{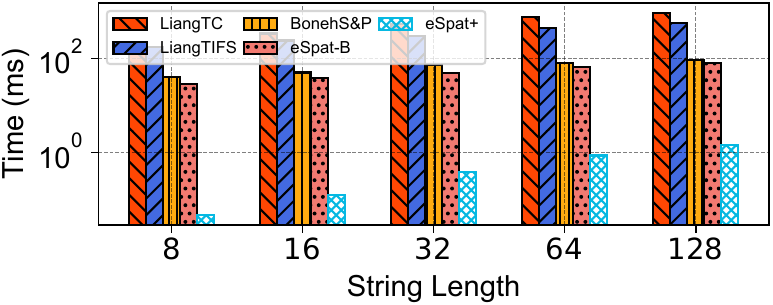}
    \vspace{-.2in}
    \label{fig:keyeval-spatial}
    \end{minipage}
    }\\
    \subfloat[][\normalsize{KeyEval time w.r.t data records.}]{
	\begin{minipage}{.3\linewidth}
    \includegraphics[width=\linewidth]{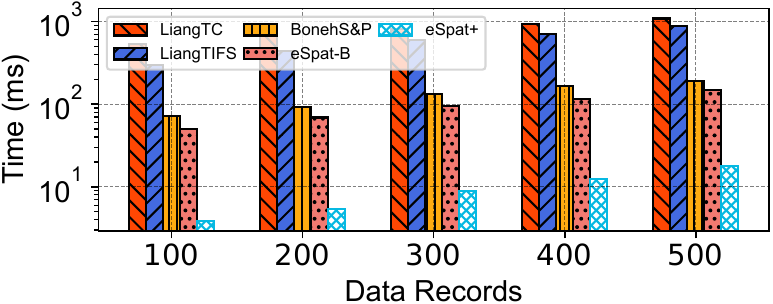}
    \vspace{-.2in}
    \label{fig:keyeval-UAV}
    \end{minipage}
    }
    \vspace{-.15in}
    \subfloat[][\normalsize{Update time w.r.t string length.}]{
	\begin{minipage}{.3\linewidth}
    \includegraphics[width=\linewidth]{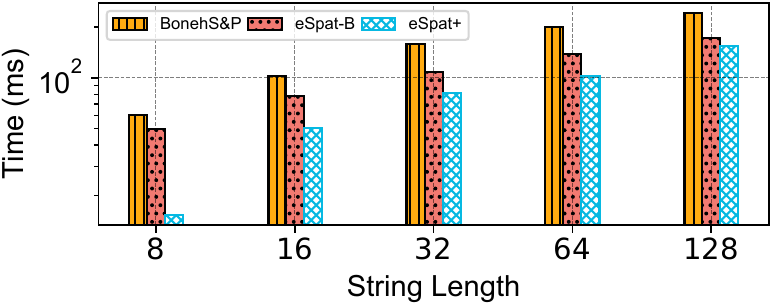}
    \vspace{-.2in}
    \label{fig:update-spatial}
    \end{minipage}
    } 
    \subfloat[][\normalsize{Update time w.r.t data records.}]{
	\begin{minipage}{.3\linewidth}
    \includegraphics[width=\linewidth]{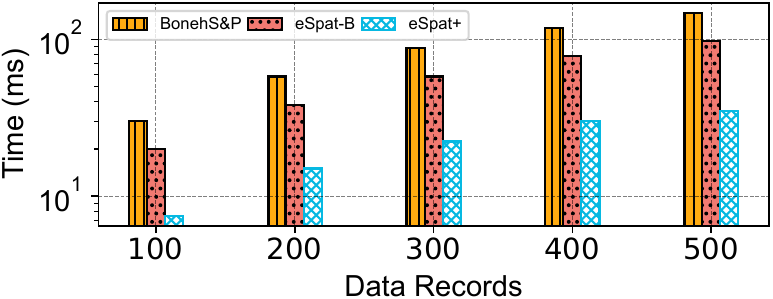}
    \vspace{-.2in}
    \label{fig:update-UAV}
    \end{minipage}
    }\\
    \subfloat[][\normalsize{Statistics time w.r.t string length.}]{
	\begin{minipage}{.3\linewidth}
    \includegraphics[width=\linewidth]{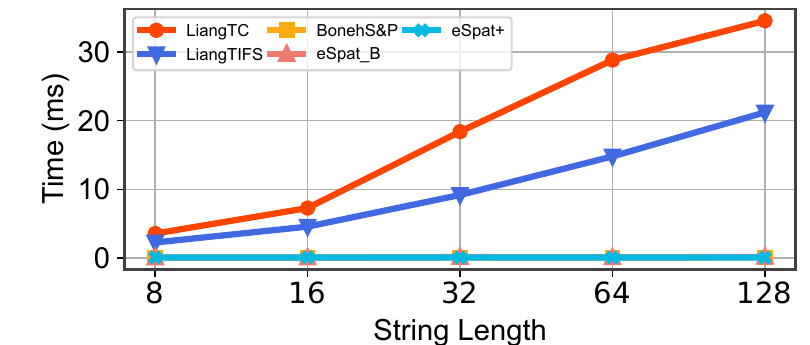}
    \vspace{-.2in}
    \label{fig:statistics-spatial}
    \end{minipage}
    }
    \vspace{-.15in}
    \subfloat[][\normalsize{Statistics time w.r.t data records.}]{
	\begin{minipage}{.3\linewidth}
    \includegraphics[width=\linewidth]{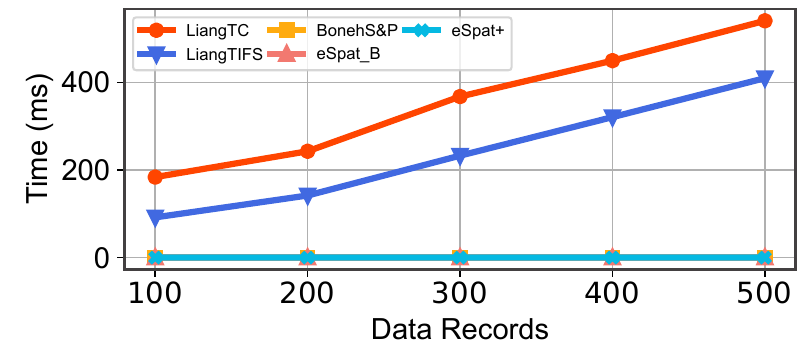}
    \vspace{-.2in}
    \label{fig:statistics-UAV}
    \end{minipage}
    }
    \subfloat[][\normalsize{Total time w.r.t string length.}]{
	\begin{minipage}{.3\linewidth}
    \includegraphics[width=\linewidth,height=2.3cm]{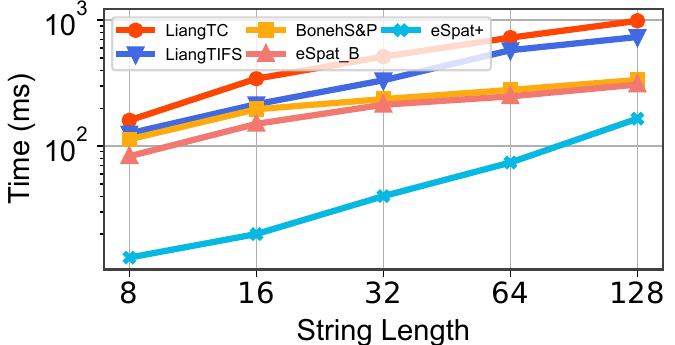}
    \vspace{-.2in}
    \label{fig:toatl_spatail}
    \end{minipage}
    }\\
    \subfloat[][\normalsize{Total time w.r.t data records.}]{
	\begin{minipage}{.3\linewidth}
    \includegraphics[width=\linewidth,height=2.3cm]{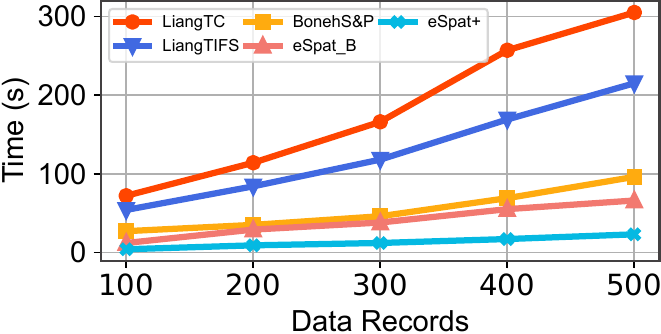}
    \vspace{-.2in}
    \label{fig:total_UAV}
    \end{minipage}
    }
    % \vspace{-0.15in}
    \subfloat[][\normalsize{Communication cost of clients.}]{
	\begin{minipage}{.3\linewidth}
    \includegraphics[width=\linewidth]{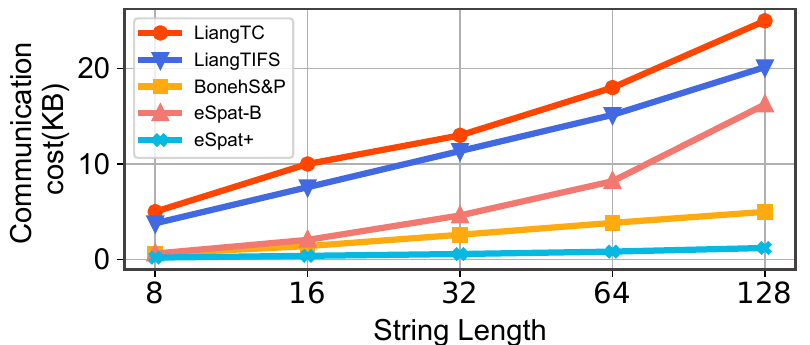}
    \vspace{-.2in}
    \label{fig:comm-clients}
    \end{minipage}
    }
    \subfloat[][\normalsize{Communication cost of servers.}]{
	\begin{minipage}{.3\linewidth}
    \includegraphics[width=\linewidth]{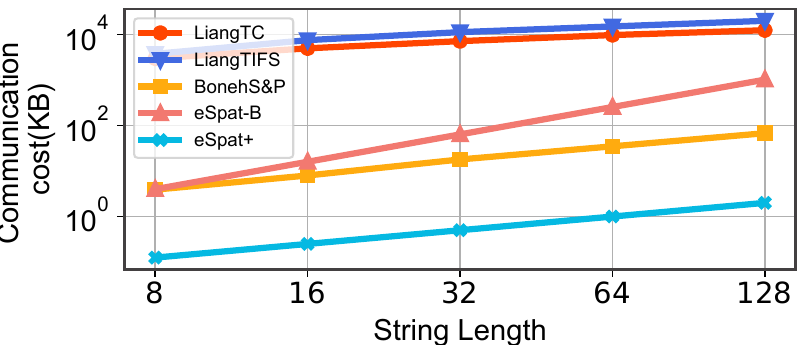}
    \vspace{-.2in}
    \label{fig:communication}
    \end{minipage}
    }
    \caption{Performance evaluation of $\mathsf{eSpat}\mbox{-}B$ and $\mathsf{eSpat+}$ against baseline schemes. Results demonstrate that $\mathsf{eSpat+}$ achieves the best efficiency in terms of computation time and communication cost for both clients and servers. For example, $\mathsf{eSpat+}$  reduces key generation time by over 60\% and client communication by over 75\% compared to strong baselines at 128-bit length.}
    % \vspace{-0.15in}
\end{figure*}

\textbf{{Metric}}. We measure the following two indicators: (1) Time costs: we evaluate the time costs of key generation, key evaluation, and update; (2) Communication costs: we evaluate the communication overhead of clients and server; (3) Accuracy: we compare the accuracy of different DP schemes.

\subsection{Experimental Results}

\textbf{Time costs for key generation.} As shown in Figs. \ref{fig:keygen-spatial} and \ref{fig:keygen-UAV}, key generation time increases with encoding string length and data records. BonehS\&P incurs higher costs due to tree height extension, LiangTIFS suffers from inefficient key construction, and LiangTC faces increased overhead as the data scale grows. In contrast, $\mathsf{eSpat\mbox{-}B}$ uses optimized DPF and octree partitioning, while $\mathsf{eSpat+}$ improves scalability with KD-tree encoding and incremental DPF. For example, at $128$ bits, $\mathsf{eSpat+}$ and $\mathsf{eSpat\mbox{-}B}$ take $160$ ms and $300$ ms, compared to $424$ ms for BonehS\&P, $987$ ms for LiangTC, and $730$ ms for LiangTIFS. Similarly, with $500$ data records, $\mathsf{eSpat+}$ and $\mathsf{eSpat\mbox{-}B}$ take $20$ and $65$ s, much less than $95$ s for BonehS\&P, $300$ s for LiangTC, and over $210$ s for LiangTIFS.

\textbf{Time costs for key evaluation.} As shown in Figs. \ref{fig:keyeval-spatial} and \ref{fig:keyeval-UAV}, key evaluation time increases with encoding string length and data records. BonehS\&P incurs higher costs due to tree traversal, LiangTIFS suffers from inefficiencies in decryption and data reconstruction operations, and LiangTC scales linearly with evaluation points. In contrast, $\mathsf{eSpat\mbox{-}B}$ reduces time using the improved DPF, while $\mathsf{eSpat+}$ further improves efficiency with KD-tree partitioning and incremental DPF. For example, at $128$ bits, $\mathsf{eSpat\mbox{-}B}$ and $\mathsf{eSpat+}$ take $80$ ms and $2$ ms, respectively, compared to over BonehS\&P ($90$ ms), LiangTC ($900$ ms), and LiangTIFS ($580$ ms). Similarly, for $500$ data records, $\mathsf{eSpat+}$ and $\mathsf{eSpat\mbox{-}B}$ take $18$ ms and $150$ ms, much lower than $190$ ms for BonehS\&P, $1000$ ms for LiangTC, and over $870$ ms for LiangTIFS.

\textbf{Time costs for update.} As shown in Figs. \ref{fig:update-spatial} and \ref{fig:update-UAV}, we evaluate the key update time with encoding string length and data records. LiangTC and LiangTIFS are excluded as they do not handle key updates. As string length and data records increase, time costs rise, but $\mathsf{eSpat\mbox{-}B}$ and $\mathsf{eSpat+}$ outperform BonehS\&P. BonehS\&P incurs higher costs due to full key regeneration during updates, while $\mathsf{eSpat\mbox{-}B}$ and $\mathsf{eSpat+}$ improve efficiency using optimized and incremental DPF. For example, at $128$ bits, $\mathsf{eSpat+}$ and $\mathsf{eSpat\mbox{-}B}$ take $150$ ms and $170$ ms, compared to over $250$ ms for BonehS\&P. For $500$ data records, $\mathsf{eSpat+}$ takes $35$ ms, much less than $100$ ms for $\mathsf{eSpat\mbox{-}B}$ and $150$ ms for BonehS\&P.

\begin{figure*}[]
  \centering % subfloat之间不能有空行，否则会被识别为换行。
  % \vspace{-.2in}
    \subfloat[][\normalsize{Accuracy on Geolife w.r.t string length.}]{
	\begin{minipage}{.4\linewidth}
    \includegraphics[width=\linewidth]{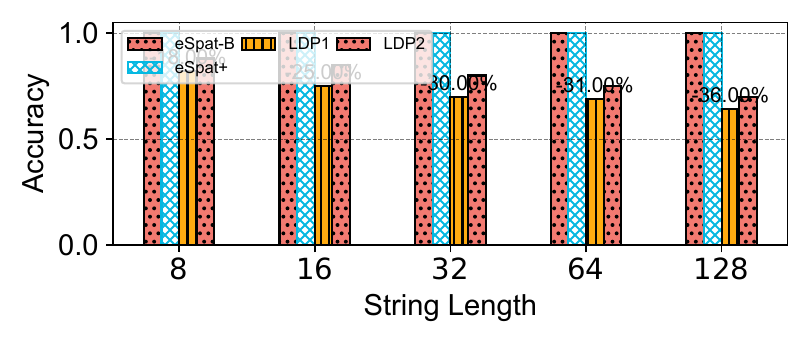}
    \vspace{-.2in}
    \label{fig:acc_string_1}
    \end{minipage}
    }
    \vspace{-0.15in}
    \subfloat[][\normalsize{Accuracy on Geolife w.r.t data records.}]{
	\begin{minipage}{.4\linewidth}
    \includegraphics[width=\linewidth,height=0.4\linewidth]{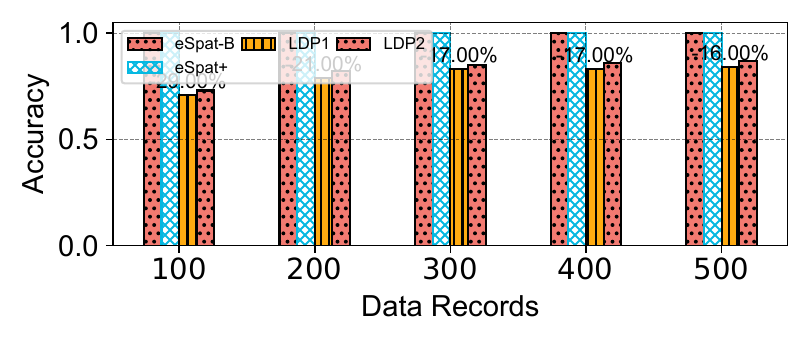}
    \vspace{-.2in}
    \label{fig:acc_UAV_1}
    \end{minipage}
    }\\ 
    \subfloat[][\normalsize{Accuracy on T-Drive w.r.t string length.}]{
	\begin{minipage}{.4\linewidth}
    \includegraphics[width=\linewidth]{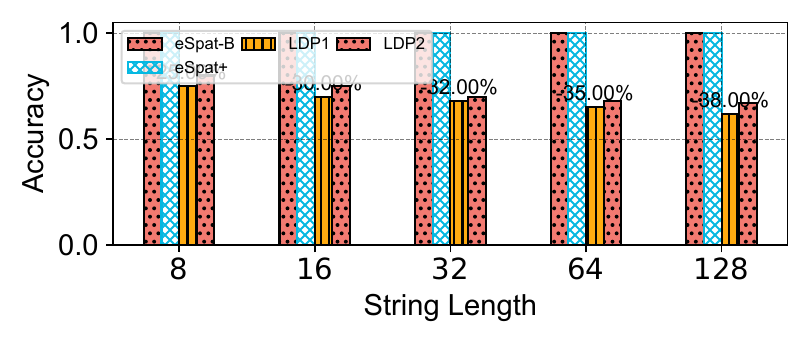}
    \vspace{-.2in}
    \label{fig:acc_string_2}
    \end{minipage}
    }
    \subfloat[][\normalsize{Accuracy on T-Drive w.r.t data records.}]{
	\begin{minipage}{.4\linewidth}
    \includegraphics[width=\linewidth]{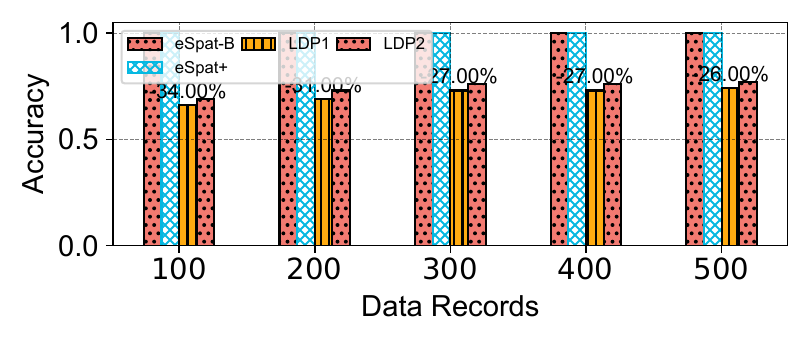}
    \vspace{-.2in}
    \label{fig:acc_UAV_2}
    \end{minipage}
    }
    \caption{Accuracy comparison of the proposed eSpat schemes against local differential privacy baselines (LDP1, LDP2) on real-world trajectory datasets. The results demonstrate that both eSpat-B and eSpat+ consistently achieve 100\% accuracy across different encoding string lengths and data volumes, as they perform exact computations without noise.}
    \label{fig:acc}
    % \vspace{0.2in}
\end{figure*}

\textcolor{black}{\textbf{Time costs for statistics.} As shown in Figs. 5g and 5h, we evaluate the statistics time, which refers to the requester-side aggregation time. As these numbers increase, time costs rise. LiangTC and LiangTIFS incur higher costs due to the need to match all indices. In contrast,  BonehS\&P only needs to add up the shares received from the server, and the same applies to $\mathsf{eSpat+}$ and $\mathsf{eSpat\mbox{-}B}$. For instance, at $128$ bits, BonehS\&P, $\mathsf{eSpat+}$ and $\mathsf{eSpat\mbox{-}B}$ take $0.1$ ms, compared to $540$ ms for LiangTC, and $410$ ms for LiangTIFS. For $500$ data records, $\mathsf{eSpat+}$ and $\mathsf{eSpat\mbox{-}B}$ also take $0.1$ ms, significantly lower than LiangTC and LiangTIFS.}

\textbf{Total time.} As shown in Figs. \ref{fig:toatl_spatail} and \ref{fig:total_UAV}, we evaluate the total time with the encoding string and data records. The reasons for the differences are discussed in earlier experiments. Specifically, at $128$ bits, $\mathsf{eSpat+}$ and $\mathsf{eSpat\mbox{-}B}$ take $160$ and $300$ ms, compared to BonehS\&P ($424$ ms), LiangTC ($987$ ms), and LiangTIFS ($730$ ms). Similarly, with $500$ data records, $\mathsf{eSpat+}$ and $\mathsf{eSpat\mbox{-}B}$ take $20$ and $65$ s, much less than  BonehS\&P ($95$ s), LiangTC ($300$ s), and LiangTIFS ($210$ s).

\textbf{Communication overhead.} As shown in Figs. \ref{fig:comm-clients} and \ref{fig:communication}, we evaluate communication overhead with different encoding string length. The communication costs of $\mathsf{eSpat\mbox{-}B}$, $\mathsf{eSpat+}$, and the baselines increase with string length. $\mathsf{eSpat+}$ achieves the lowest communication cost due to incremental DPF with KD-tree partitioning. In contrast, $\mathsf{eSpat\mbox{-}B}$ incurs quadratic costs due to linear-sized keys at each prefix tree level, while BonehS\&P, leveraging incremental DPF, incurs higher costs due to greater tree height, though still lower than $\mathsf{eSpat\mbox{-}B}$. LiangTC and LiangTIFS need to generate binary trees and keys for the servers, leading to the highest communication costs. For instance, at $128$ bits, $\mathsf{eSpat+}$ incurs $1.2$ kilobyte (KB) for client communication and $7$ KB for server communication, significantly lower than $5$ KB and $70$ KB for BonehS\&P, $16$ KB and $1000$ KB for $\mathsf{eSpat\mbox{-}B}$, and over $20$ KB and $20000$ KB for LiangTC and LiangTIFS.

\textcolor{black}{\textbf{Accuracy.} As shown in Fig.~\ref{fig:acc}, we compare the accuracy under different encoding string lengths and numbers of data records. LDP1 and LDP2 denote representative local differential privacy schemes. The results show that $\mathsf{eSpat\mbox{-}B}$ and $\mathsf{eSpat+}$ consistently achieve 100\% accuracy, whereas LDP1 and LDP2 exhibit lower accuracy due to the perturbation noise introduced for privacy protection. The proposed schemes perform exact cryptographic computation: each spatial point is encoded into a spatial index, the two servers evaluate DPF/IDPF shares, and the requester reconstructs the exact counting result without introducing noise or approximation. In contrast, the accuracy of LDP-based approaches decreases as the spatial partition becomes finer, since perturbation may move data points across region boundaries. As the number of data records increases, the impact of noise is averaged out, leading to gradually improved accuracy. }

\section{Conclusion}\label{sec:8}
In this paper, we have proposed efficient and privacy-preserving distribution statistics analysis schemes for spatial data, which achieve accurate statistics while preserving spatial data privacy. First, we have designed $\mathsf{eSpat\mbox{-}B}$, which combines distributed point functions and octree partitioning to guarantee both privacy and efficiency. Building upon this, we have introduced $\mathsf{eSpat+}$, which further reduces computational and communication overheads by leveraging KD-tree and incremental DPF. In addition, we have also designed an efficient update algorithm to process the frequent updates of spatial data. Security analysis and experimental results have validated the robustness and high performance of $\mathsf{eSpat\mbox{-}B}$ and $\mathsf{eSpat+}$.


\bibliographystyle{IEEEtran}
\bibliography{ref.bib}

\newpage

\vfill

\end{document}